\documentclass[sigconf]{acmart}
\usepackage{balance}

\usepackage{algorithmic}
\usepackage[ruled,linesnumbered]{algorithm2e} 
\usepackage[titletoc]{appendix} 
\usepackage{booktabs} 
\usepackage{bm} 
\usepackage{caption} 
\usepackage{diagbox} 
\usepackage{float} 
\usepackage{graphicx} 
\graphicspath{{./figures/}} 

\usepackage{listings} 
\usepackage{mathrsfs} 
\usepackage{mathtools} 
\usepackage{multirow} 

\usepackage{subcaption} 
\usepackage{thmtools} 
\usepackage{thm-restate} 
\usepackage{tikz} 

\usepackage[textsize=normalsize]{todonotes} 
\usepackage{xcolor} 

\definecolor{pku-red}{RGB}{139,0,18} 
\usepackage[capitalize,noabbrev]{cleveref} 

\theoremstyle{plain} 
\newtheorem{theorem}{Theorem}[section]

\newtheorem{corollary}[theorem]{Corollary}
\newtheorem{example}{Example}

\theoremstyle{definition}
\newtheorem{definition}{Definition}[section]

\theoremstyle{remark}
\newtheorem{remark}{Remark}[section]
\theoremstyle{remark}
\newtheorem*{remark*}{Proof Sketch and Remark}

\newcommand{\bbE}{\mathbb{E}}

\newcommand{\bbN}{\mathbb{N}}

\newcommand{\bbR}{\mathbb{R}}

\newcommand{\bbZ}{\mathbb{Z}}

\newcommand{\bmb}{{\bm{b}}}
\newcommand{\bmc}{{\bm{c}}}
\newcommand{\bmd}{{\bm{d}}}

\newcommand{\bmf}{{\bm{f}}}

\newcommand{\bmk}{{\bm{k}}}

\newcommand{\bmu}{{\bm{u}}}
\newcommand{\bmv}{{\bm{v}}}
\newcommand{\bmw}{{\bm{w}}}
\newcommand{\bmx}{{\bm{x}}}
\newcommand{\bmy}{{\bm{y}}}
\newcommand{\bmz}{{\bm{z}}}

\newcommand{\bmaa}{\bm{\alpha}}

\newcommand{\bmdd}{\bm{\delta}}

\newcommand{\bmgg}{\bm{\gamma}}

\newcommand{\bfx}{\mathbf{x}}

\newcommand{\iinn}{{i\in[n]}} 
\newcommand{\jinn}{{j\in[n]}}
\newcommand{\kinn}{{k\in[n]}}

\newcommand{\ijinn}{{i,j\in[n]}}

\newcommand{\jnei}{{j\ne i}}

\newcommand{\knei}{{k\ne i}}

\newcommand{\knej}{{k\ne j}}

\newcommand{\one}{\bm{1}} 
\newcommand{\ones}{\mathbf{1}} 
\newcommand{\zeros}{\mathbf{0}} 
\DeclareMathOperator*{\argmax}{arg\,max}

\newcommand{\dd}{\mathrm{d}} 
\newcommand{\pp}{\partial} 
\newcommand{\Var}{\mathrm{Var}} 
\newcommand{\diag}{\mathrm{diag}} 

\newcommand{\ie}{\emph{i.e.}}    
\newcommand{\eg}{\emph{e.g.}}    
\newcommand{\wrt}{\emph{w.r.t.}}  
\newcommand{\iid}{\mathrm{i.i.d.}}  

\newcommand{\SW}{\mathrm{SW}}
\newcommand{\xhigh}{\bar{x}}
\newcommand{\xlow}{\underline{x}}
\newcommand{\khigh}{\bar{k}}
\newcommand{\klow}{\underline{k}}
\newcommand{\dhigh}{\bar{d}}
\newcommand{\dlow}{\underline{d}}
\newcommand{\BR}{\mathrm{BR}}

\newcommand\myx[1]{\textcolor{cyan}{[MYX: #1]}}

\renewcommand\myx[1]{}

\AtBeginDocument{%
  }

\setcopyright{acmlicensed}
\copyrightyear{2025}
\acmYear{2025}
\acmDOI{10.1145/3696410.3714869}

\acmConference[WWW '25] {Proceedings of the ACM Web Conference 2025}{April 28--May 2, 2025}{Sydney, NSW, Australia.}
\acmISBN{979-8-4007-1274-6/25/04}

\begin{document}

\title{Networked Digital Public Goods Games with \texorpdfstring{\\}{ } Heterogeneous Players and Convex Costs}


\author{Yukun Cheng}
\email{ykcheng@amss.ac.cn}
\affiliation{
    \institution{School of Business \\ Jiangnan University}
    \country{Wuxi, China}
}

\author{Xiaotie Deng}
\authornote{Corresponding author}
\email{xiaotie@pku.edu.cn}
\affiliation{
    \institution{CFCS, School of Computing \\ CMAR Institute for Al \\ Peking University}
    \country{Beijing, China}
}

\author{Yunxuan Ma}
\authornote{Main contributor to this research.}
\email{yunxuanma@pku.edu.cn}
\affiliation{
    \institution{CFCS, School of Computing \\ Peking University}
    \country{Beijing, China}
}







\renewcommand{\shortauthors}{Yukun Cheng, Xiaotie Deng, \& Yunxuan Ma}


\def\comment#1{}
\begin{abstract}

In the digital age, resources such as open-source software and publicly accessible databases form a crucial category of digital public goods, providing extensive benefits for Internet. 
However, these public goods' inherent non-exclusivity and non-competitiveness frequently result in under-provision, a dilemma exacerbated by individuals' tendency to free-ride. This scenario fosters both cooperation and competition among users, leading to the public goods games.

This paper investigates networked public goods games involving heterogeneous players and convex costs, focusing on the characterization of Nash Equilibrium (NE). In these games, each player can choose her effort level, representing her contributions to public goods.
Network structures  are employed to model the interactions among participants. Each player's utility consists of a concave value component, influenced by the collective efforts of all players, and a convex cost component, determined solely by the individual's own effort. 
To the best of our knowledge, this study is the first to explore the networked public goods game with convex costs.


Our research begins by examining welfare solutions aimed at maximizing social welfare and ensuring the convergence of pseudo-gradient ascent dynamics. We establish the presence of NE in this model and provide an in-depth analysis of the conditions under which NE is unique. 
We also delve into \emph{comparative statics}, an essential tool in economics, to evaluate how slight modifications in the model—interpreted as monetary redistribution—affect player utilities. In addition, we analyze a particular scenario with a predefined game structure, illustrating the practical relevance of our theoretical insights.
Overall, our research enhances the broader understanding of strategic interactions and structural dynamics in networked public goods games, with significant implications for policy design in internet economic and social networks.

\end{abstract}






\maketitle

\section{Introduction}
\label{sec:intro}

The concept of public goods is not only a significant area of interest in economic research but also closely related to the web era.
Public goods encompass a wide range of web resources, in forms like open-source software (\eg, GitHub), public databases (\eg, MNIST), scientific technologies (\eg, papers in The WebConf), and widely accessible scientific knowledge (\eg, Wikipedia, Stack overflow).
The defining characteristics of public goods are their non-excludability, meaning all community members can freely use these resources without excluding anyone, and non-rivalry, where one person’s use does not diminish the availability for others. 
Such characteristics are particularly notable on the internet. 
Web and internet research delves into how to effectively provide and manage digital public goods to maximize social welfare. This exploration is not just theoretical but also has practical implications for policy and the development of website content, attracting an increasing number of researchers to this burgeoning field.


However, from a societal perspective, digital public goods often face challenges due to insufficient provision, a problem frequently attributed to the issue of free-riding. Consequently, each participant must decide how much effort to contribute when investing in digital public goods, aware that their efforts will also benefit others. This strategic decision-making process embodies what is known as a {\it public goods game}. This game can reveal complex interactions between cooperation and competition, as individuals shall balance their contributions against the collective benefits. Much of the prior research~\cite{public-network:bramoulle2007public} has focused on idealized models where participants are assumed to be homogeneous. However, in reality, especially in the case of digital public goods, users exhibit significant heterogeneity. For example, a specialized dictionary on Wikipedia is more beneficial to those within the relevant field. On the other hand, in the context of paper reviews, the efforts of one reviewer benefit the entire conference but may disadvantage the author of a low-quality submission. This demonstrates that the impact of a public good (or bad) can be either positive or negative, and varies across different participants.

In this paper, we are more interested in the networked public goods games, which effectively capture the social connections 
among individuals. Specifically, all participants are positioned at the vertices of the network, and the links—each weighted differently—represent the relationships and influence between any two participants~\cite{jichen:li2023altruism}.
\citet{public-network:bramoulle2007public} pioneered the study of public goods games within a network. 
In their homogeneous model, the utility functions of all players are consistently formulated as $u_i(\bmx) = f(x_i + \sum_{j \in N_i} x_j) - cx_i$, where $x_i$ is the effort level of player $i$, $f(\cdot)$ is a homogeneous benefit function applicable to all players, and the cost function is linear, characterized by a uniform unit cost $c$ for all players. 
Furthermore, this model is unweighted, as each player exhibits the same preference for both their efforts and those of others when computing the benefit.
Based on this simplified and idealistic setting, \citet{public-network:bramoulle2007public} demonstrated the existence of an equilibrium where some players exert the same maximum effort while all others engage in free riding. Moreover, they showed that those contributing positive effort form an independent set within the network. 
While later studies have explored the public goods games with heterogeneous utility functions \cite{public-network-direct-complexity:papadimitriou2021public,public-network-direct-BRD:bayer2023best}, their focus remained on the linear cost scenarios.

However, practical scenarios often feature non-linear cost functions, particularly evident in digital public goods.
For instance, the initial setup of a Wikipedia article involves adding basic facts and general information—tasks that are relatively low in cost. Yet, as the article develops, ensuring accuracy and providing in-depth analysis demand increasingly specialized knowledge, research, and citations, raising the marginal cost of contributions. Unfortunately, the predominant body of research on public goods games focuses on linear cost functions \cite{public-network:bramoulle2007public,public-network-direct:lopez2013public,public-network-direct-complexity:papadimitriou2021public,public-network-direct-BRD:bayer2023best},  and very few studies delve into the implications of non-linear cost functions.

This paper presents a novel model of networked public goods games that incorporates convex cost functions, aiming at understanding the equilibrium and dynamics in the field of digital public goods. 
Specifically, given an effort profile $\bfx=(x_1,x_2,\cdots,x_n)$, each player's payoff is determined by the net gain, which is the difference between a benefit function $f_i(k_i)$ and a cost function $c_i(x_i)$. The benefit function $f_i$ for player $i$ is both concave and strictly increasing, and it is derived from the gain $k_i$. This gain $k_i$ is computed as a weighted linear combination of the efforts of both the player and her neighbors. The cost function $c_i$, which is convex and strictly increases, depends exclusively on the player's effort $x_i$.


\subsection{Results and Techniques}
Our work is the first one to study the networked public good games with convex costs. The heterogeneity of benefit functions and cost functions lends greater generality to the networked public goods game studied in this paper.
We start by exploring the concept of welfare solutions, focusing on the maximization of social welfare and the investigation of pseudo-gradient ascent dynamics, which shows insight into the following analysis. We carefully analyze the existence and uniqueness of Nash Equilibrium (NE) across various settings, providing deep insights into the NE's structures in public goods games. Our examination extends to cases in which distinct characteristics of cost and benefit functions play a crucial role in ensuring the NE's uniqueness. Building on these foundations, we delve into comparative statics to assess the effects of subtle shifts in the model's parameters, which we regard as money redistribution, on the utilities of the players involved. Comparative statics is a crucial analytical method in economics. This element of our study illuminates how minor adjustments can significantly influence economic outcomes and player behaviors within the game.
We also study a special case, in which the game structure is pre-defined and show how these theorems can be applied to this case.

The proof of the existence of a Nash Equilibrium (NE) primarily relies on the application of the Brouwer fixed-point theorem.
Brouwer's fixed-point theorem states that any continuous function mapping a compact, convex set to itself must have a fixed point \cite{brouwer:brouwer1911abbildung}. 
It's important to note that the best-response function is continuous when the utility functions are strictly concave. The proof then proceeds through a strategic modification of the utility functions, ensuring they meet the criteria stipulated by Brouwer's fixed-point theorem.

To carve out the uniqueness of NE, we bring out the concept of near-potential game and show that under certain conditions, the NE of near-potential game is unique, and pseudo-gradient ascent dynamic will converge to this point with exponential rate. The proof constructs the discrete version of pseudo-gradient ascent dynamic and shows that it is compressive mapping, which is guaranteed to have a unique fixed point by Banach's theorem \citep{banach:banach1922operations}.
We then bridge the gap between near-potential games and public good games, showing three conditions under which we can transform the public good games into a specifically designed near-potential game while holding the NEs invariant, therefore guaranteeing the uniqueness of NE. 
We also propose the concept of game equivalence, which ensures the one-to-one relationships between the NEs of corresponding games, which can also broaden the class of games possessing unique NE.

To study the comparative statics on money redistribution, we mainly use the high-dimensional implicit function theorem.
We rewrite the conditions of Nash equilibrium $\bmx^*$ as an implicit function of infinitesimal model change $t\bmdd$ and corresponding NE $\bmx^*(t)$.
By differentiating this implicit function on $t$, we can derive the relation between $\bmx^*(t)$, $\bmdd$ and $t$.

\subsection{Related Works}

\emph{\bf Public Goods in the Web Era.}
In the web era, public goods play a crucial role in fostering collaborative contributions and maintaining online platforms. 
\citet{public_good-web-1:gallus2017fostering} demonstrates the impact of symbolic awards on volunteer retention in a public good setting like Wikipedia, where recognition and community engagement can encourage sustained contributions without direct financial incentives. 
Similarly, the challenges of knowledge-sharing in Web 2.0 communities have been framed as a public goods problem, where social dilemmas like free-riding are mitigated through enhanced group identity and pro-social behavior \citep{public_good-web-2:allen2010knowledge}.
Experimental research on cooperation in web-based public goods games further examines how network structures influence contribution behavior, with findings suggesting that contagion effects in cooperative behavior are limited to direct network neighbors \citep{public_good-web-3:suri2011cooperation}.
Moreover, the broader economic dynamics of the web are analyzed through the concept of "web goods," where users contribute content, exchange information, and interact in a socio-economic system that requires balancing open access with incentive structures for content production and infrastructure development \citep{public_good-web-4:vafopoulos2012web}.
These works collectively highlight the unique challenges and opportunities of managing public goods in the digital age, emphasizing the importance of community-driven incentives and network effects in fostering web-based cooperation.

\emph{\bf Networked Public Good Games.}
\citet{public-network:bramoulle2007public} initiated the study of public goods in a network. They studied the public good games on an unweighted, undirected network with linear cost functions and homogeneous players. Under their models, there is a unique level $k^*$ such that it's optimal for each player to make the sum efforts within her neighborhoods to be $k^*$, which greatly simplifies the analysis of the model. The authors showed that the NE of the game corresponds to the maximal independent set, where the player in the maximal independent set asserts full effort $k^*$, and the players outside free-ride. 

There are many other works following this literature, see \citep{public-network-follow:bramoulle2014strategic, public-network-follow:boncinelli2012stochastic,public-network:allouch2015private,public-network:elliott2021network}.
\citet{public-network-follow:bramoulle2014strategic} extended the model to the imperfectly substituted
public goods case, and proved the existence and uniqueness of Nash equilibrium, under the condition of sufficiently small lowest eigenvalue of the graph matrix.
\citet{public-network:allouch2015private} differentiated the provision of public goods and private goods, and their results of existence and uniqueness of Nash equilibrium also rely on the lowest eigenvalue of the graph matrix.
\citet{public-network-direct:lopez2013public} began with the studies of public good games in directed networks, by discussing both the static model and the dynamic model. To be specific, in the static model, all players are situated within a fixed network where they choose their actions simultaneously. \citet{public-network-direct:lopez2013public} demonstrated that the structure of Nash equilibria correlates with the maximal independent set. In contrast, the dynamic model is characterized by a dynamic sampling process, where agents periodically sample a subset of other agents and base their decisions on a myopic-best response. The author established the existence of a unique globally stable proportion of public good providers in this model. \citet{public-network-direct-BRD:bayer2023best} studied the convergence of best response dynamic on the public good games in directed networks.

A significant networked public goods game variant considers indivisible goods, where players can only make binary decisions\cite{galeotti2010network}. 
Building upon this binary networked public goods (BNPG) game model, \citet{yu2020computing} introduced the algorithmic inquiry of determining the existence of pure-strategy Nash equilibrium (PSNE). 
Specifically, they investigated the existence of PSNE in the BNPG game and proved that it is NP-hard in both homogeneous and heterogeneous settings.
The computational complexity of public goods games with a network structure, such as tree or clique~\cite{yang2020refined, maiti2024parameterized}, and regular graph~\cite{feldman2013pricing} has also been extensively studied.
Papadimitriou and Peng\cite{public-network-direct-complexity:papadimitriou2021public} proved that finding an approximate NE of the public good games in directed networks is PPAD-hard, even if the utility is in a summation form.
Subsequently, \citet{public-uncertainty:gilboa2022complexity} modeled players as different patterns and showed that the existence of PSNE on some non-trivial patterns is NP-complete, while a polynomial time algorithm exists for some specific patterns.
In addition, \citet{klimm2023complexity} further demonstrated the complexity results of the BNPG game on undirected graphs with different utility patterns to be NP-hard. 
They also showed that computing equilibrium in games with integer weight edges is PLS-complete.

\emph{\bf Continuous-time Public Good Games.}
A branch of the literature on public goods focuses on studying the dynamic provision of public goods in continuous time. \citet{public-dynamic:fershtman1991dynamic} was the first to explore this problem. They proposed two equilibrium concepts: the open-loop equilibrium and the feedback equilibrium, showing that in the feedback equilibrium, the players' utilities are lower than in the open-loop equilibrium. This result is derived under the linear strategy assumption of the feedback equilibrium, as the feedback equilibrium is not generally unique. Later, \citet{public-dynamic-nonlinear:wirl1996dynamic} discovered that if non-linear strategies are allowed in the feedback equilibrium, players' utilities can be higher in some feedback equilibria than in the open-loop equilibrium. \citet{public-dynamic-CES:fujiwara2009dynamic} generalized these findings to more general utility functions and confirmed that the results still hold. \citet{public-dynamic-uncertain:wang2010dynamic} extended this work by considering environments with uncertainty.

Although these studies present findings in dynamic scenarios, they generally assume homogeneity among players in terms of utility functions (both gains and costs) and interpersonal relationships and thus do not consider network effects. To the best of our knowledge, no previous research has simultaneously explored the dynamic provision of public goods with heterogeneous players.

\emph{\bf Concave Games.}
\citet{concave_game-initial:rosen1965existence}
firstly introduced the concept of concave games, in which the utility function of each player is concave to her strategy. In this paper, \citet{concave_game-initial:rosen1965existence} provided a sufficient condition for such games to have a unique equilibrium and introduced a differential equation that converges to this equilibrium. 
Because of the foundational results of \citet{concave_game-initial:rosen1965existence}, several works have extended the study of concave games in various settings, such as learning perspective of equilibrium in concave games \citep{concave_game-learning:nesterov2009primal, concave_game-learning:mertikopoulos2019learning, concave_game-learning:bravo2018bandit}, equilibrium concept in concave games \citep{concave_game-concept:forgo1994existence, concave_game-concept:ui2008correlated, concave_game-concept:goktas2021convex}.
However, limited research has applied the concave game framework to public goods scenarios. Our work is pioneering in applying the convex game framework to public goods games. We demonstrate that public goods games can be treated as a specific type of concave game, called a near-potential game, where the potential function is meticulously designed for diverse scenarios. The uniqueness of equilibrium in near-potential games, therefore, directly supports the uniqueness of equilibrium in public goods games.

\section{Models}
\label{sec:model}

Consider a community with $n$ players playing a public good game.
Each player $i$ needs to decide her effort $x_i \in [\underline{x}_i, \bar{x}_i] \coloneqq X_i$ to invest the public goods, where $\{\underline{x}_i, \bar{x}_i\}_\iinn$ are predetermined and publicly known. 
Let $\bmx = (x_1,...,x_n)$ be the effort profile of all players, and $\bmx_{-i}$ be the effort profile of all players without player $i$. 
Therefore, $(y_i,\bmx_{-i})$ is the effort profile that player $i$ chooses $y_i$ and other players keep their choices the same as $\bmx_{-i}$.
Similarly, define $X = \times_\iinn X_i$ and $X_{-i} = \times_\jnei X_j$.
Let us denote $W=\{w_{ij}\}_{i,j\in [n]}$
as the matrix, in which $w_{ij}$ represents the marginal gain of player $i$
from player $j$'s effort. We normalize $W$
such that $w_{ii}=1$ for any $i\in [n]$
without loss of generality. Our model is more general as it imposes no additional constraints on the network, \eg, $w_{ij}\in [0,1]$ or $w_{ij} = w_{ji}$. 
Any $w_{ij}\in \bbR$ are permitted, provided that $w_{ij}=0$ if there is no edge between $i$ and $j$.
Let $k_i\in K_i$ be the total gain of player $i$ and $\bmk = (k_1,...,k_n) \in K$ be the gain profile of all players. Then we have $k_i = \sum_\jinn w_{ij} x_j$, \ie, the gain of player $i$ linearly depends on her own and other players' efforts, weighted by $\bmw_{i}=(w_{ij})_{j\in [n]}$. Therefore, $\bmk = W \bmx$. 
In addition, we assume $K_i = [\klow_i, \khigh_i]$, where $\klow_i$ and $\khigh_i$ are the minimum and maximum possible gain for player $i$ for ease of representation, respectively\footnote{Since $X_i$ is bounded for all $i$, $\klow_i$ and $\khigh_i$ are well-defined. In fact, we have the explicit expression that $\klow_i = \sum_\jinn \one\{w_{ij} > 0 \} w_{ij}\xlow_j + \one\{w_{ij} < 0 \} w_{ij} \xhigh_j$ and $\khigh_i = \sum_\jinn \one\{w_{ij} > 0 \} w_{ij} \xhigh_j + \one\{w_{ij} < 0 \} w_{ij} \xlow_j$}.
Similarly, we use $K = \times_\iinn K_i$ and define $K_{-i} = \times_\jnei K_j$.

Given an effort profile ${\bmx}=(x_1,x_2,\cdots,x_n)$, each player $i$ has utility function $u_i(\bmx) = f_i(k_i) - c_i(x_i)$, where $f_i$ is a concave and strictly increasing function on $K_i$, and $c_i$ is a convex and strictly increasing function on $X_i$. That is, $f''_i(x)\le 0, f'_i(x)>0$, $c'_i(x) >0, c''_i(x) \ge 0$, meaning that $f_i$ and $c_i$ are twice differentiable. 
Thus, a networked public goods game $G$ is represented by a four-tuple, $G = \langle \{f_i\}_\iinn, \{c_i\}_\iinn, \{X_i\}_\iinn, W \rangle$, where $K_i$ is omitted since it can be uniquely determined from $G$.

The utility function  $u_i(\bmx) = f_i(k_i) - c_i(x_i)$ indicates that each player's utility is composed of two parts, the value part $f_i(k_i)$ and the cost part $c_i(x_i)$. Clearly, the value part depends on her gains $k_i$, which are positively or negatively affected by other players' efforts and the cost part only depends on her effort. Notice that one's effort will increase or decrease others' gains, and so will their utilities. Therefore players' efforts can be regarded as public goods (bads).

For the sake of convenience, we use $\bbR_+$ and  $\bbR_{++}$ to denote the set of non-negative real numbers and the set of (strict) positive real numbers, respectively.

To begin with, we introduce some definitions, which are useful for the following analysis.

\begin{definition}[$\alpha$-Lipschitzness]
\label{def:Lipschitz}
A function $g(x): X \to \bbR$ is $\alpha$-Lipschitz ($\alpha \in \bbR_+$) on $x \in X \subseteq \bbR^d$, if
\begin{align*}
    |g(x) - g(y)| \le \alpha \| x - y \|
\end{align*}
for all $x,y \in X$.
\end{definition}

\begin{definition}[$c$-concavity]
\label{def:concave}
A differentiable function $g(x): X \to \bbR$ is $c$-(strongly) concave ($c \in \bbR_+$) on $x \in X \subseteq \bbR^d$, if $X$ is a convex set and,
\begin{equation}
\label{eq:concave}
    g(y) \le g(x) + \langle y-x, \nabla g(x)\rangle - \frac{c}{2}\| y - x \|^2, \quad \forall x,y\in X.
\end{equation}
\end{definition}

Intuitively, we may understand the definition to be that $g(\cdot)$ has a directional curvature
less than or equal to $-c$ at any point $x$ inside the convex set $X$ to any direction $y - x$.

\begin{definition}[$\bmaa$-scaled Pseudo-Gradient Ascend Dynamic]
\label{def:BRD}
Let $\{u_i(\bmx)\}_\iinn$ be the utility functions of players in an $n$-player game and $\bmx(0)$ be an arbitrary initial strategy profile. 
Consider a dynamic of players' strategies, called $\bmaa$-scaled pseudo-gradient ascent dynamic, which describes the players' behaviors over time.
The $\bmaa$-scaled pseudo-gradient ascent dynamic
$\bmx(t) = (x_1(t),\dots, x_n(t))$
with updating speed $\bmaa \in \bbR_{++}^n$ (possibly $\bmaa \ne \ones$) is a system of differential equations, defined as
\begin{align*}
    \frac{\dd x_i}{\dd t}(t)
    = \alpha_i \frac{\partial u_i}{\partial x_i}
    (\bmx(t))
    ,\quad \forall \iinn.
\end{align*}
    
\end{definition}

Here the vector $(\frac{\pp u_i}{\pp x_i}(\bmx))_\iinn$ is called the pseudo-gradient of the game $(u_i(\bmx))_\iinn$ at the point $\bmx$. 
In this dynamic, each player updates her strategy, taking the direction as the gradient of her utility, scaled by vector $\bmaa$.

 
We finally present a property of $c$-concavity function for use in the subsequent theorems.

\begin{restatable}{lemma}{lemConcave}
\label{lem:concave}
Assume $g: X \to \bbR$, where $X\subseteq \bbR^d$ is a convex and closed set, and $g$ is a differential $c_0$-concave function. Define $x^*$ be the maximum point of $g(x)$ on $X$, then,
\begin{align*}
    2 c_0 \left( g(x^*) - g(x) \right) \le \| \nabla g(x) \|^2\quad \forall x\in X.
\end{align*}
\end{restatable}

For completeness, we only present self-contained proofs for all results in this paper. The remaining proofs are deferred to the Appendix, due to the space limits.

\subsection{Welfare Solutions}
\label{subsec:welfare}
We first consider the social optimal solution. This concept is characterized by the social welfare $\SW(\bmx)$, which is defined as the sum of all players' utilities: $\SW(\bmx) = \sum_\iinn (f_i(k_i) - c_i(x_i))$. The social optimal solution is the effort profile $\bmx^*$ that maximizes $\SW(\bmx)$.

Because of the concavity of the social welfare function and the convex, bounded and closed domain $X$, the existence of the social optimal solution is guaranteed.
However, the social optimal solution may not always have an explicit expression. This limitation motivates us to explore the gradient flow as a dynamic process to achieve the social optimal solution: 
\begin{equation}
\label{eq:SW:BRD}
    \frac{\dd x_i}{\dd t}(t) = \frac{\partial \SW}{\partial x_i}(\bmx(t)),\quad \iinn.
\end{equation}
It is well-established that the gradient flow converges to a stable point, and in the case of a concave function, any stable point corresponds to a global maximum. Specifically, we have the following theorem:
\begin{restatable}{theorem}{thmSWBRD}
\label{thm:SW:BRD}
    The best-response dynamic \cref{eq:SW:BRD} converges to the social optimal solution with linear rate, \ie,
    \begin{align*}
        SW(\bmx^*) - SW(\bmx(t)) \le \frac{c}{t},\quad \forall t>0
    \end{align*}
    for some $c > 0$.

    Moreover, if at least one of the following conditions holds:
    \begin{itemize}
        \item[(1)] all cost functions $c_i(x)$ are $c_0$-convex for some $c_0>0$;
        \item[(2)] all value functions $f_i(k)$ are $c_0$-concave for some $c_0>0$;
    \end{itemize}
    then, the best-response dynamic converges to the social optimal solution with exponential rate, \ie,
    \begin{align*}
        SW(\bmx^*) - SW(\bmx(t)) = O(\exp(-c\cdot t))
    \end{align*}
    for some $c>0$.
\end{restatable}

Theorem \ref{thm:SW:BRD} establishes that the pseudo-gradient ascent dynamic with homogeneous utility function ($\SW(\bmx)$ in this case), regarded as a continuous-time algorithm, converges to the social optimal point. 
While this result may not be surprising, the technical insight in this proof is helpful for later proofs of the uniqueness of NE.

\section{Equilibrium Solutions of Public Good Games}
\label{sec:equilibrium}

In this section, we establish the existence and uniqueness results of equilibrium solutions in public good games. An effort profile is $\bmx$ as an (pure strategy) NE, if no player can unilaterally increase her utility by changing her effort. Formally, $\bmx$ is an NE if, for any player $i$ and any alternative effort $x_i'$, we have
\begin{align}
\label{eq:static:NE}
    u_i(x_i', \bmx_{-i}) \le u_i(x_i, \bmx_{-i}) \quad \forall x'_i \in [\underline{x}_i, \bar{x}_i].
\end{align}


\subsection{Existence of Nash Equilibrium}

Generally, an (pure strategy) NE may not exist in normal-form games. Fortunately, the following theorem states that an NE always exists in the networked public good games studied in this paper. 

\begin{restatable}{theorem}{thmNEExist}
\label{thm:NE:exist}
    In the public good game $G = (\{f_i\}_\iinn, \{c_i\}_\iinn,$ $\{X_i\}_\iinn, W)$, an (pure strategy) NE always exists.
\end{restatable}
\begin{proof}[Proof Sketch of \cref{thm:NE:exist}]
Let us construct the best-response function $\BR(\bmx)$ for all players. If it is continuous, then by Brouwer's fixed point theorem \citep{brouwer:brouwer1911abbildung}, there is a fixed point, which is also the NE of the game.

However, this is not always the case, since the best response may be discontinuous and even not a singleton set. To address this issue, we modify the cost function to be an $\alpha$-convex function. This results in an $\alpha$-modified game, whose best-response function is continuous.

As long as the $\alpha$-modified game has an NE $\bmx^*_\alpha$, we first let $\alpha \to 0$, then by compactness of $X$, we claim that there is an accumulation point $\bmx^*$ that is the limitation of $\bmx^*_{\alpha_k}$ for a sequence $\alpha_k\to 0$. 
To establish the desirable result in this theorem, we continue proving that $\bmx^*$ is the NE of the original game, given $\bmx^*_{\alpha_k}$ is the NE of $\alpha_k$-modified game for all $k$.
This involves carefully taking limits through several steps.
\end{proof}

Similar to normal-form games, NE in public goods games may not be unique. This is elaborated upon in the example below.

\begin{figure}[t]
\centering
\begin{subfigure}{0.3\textwidth}
    \includegraphics[width=\textwidth]{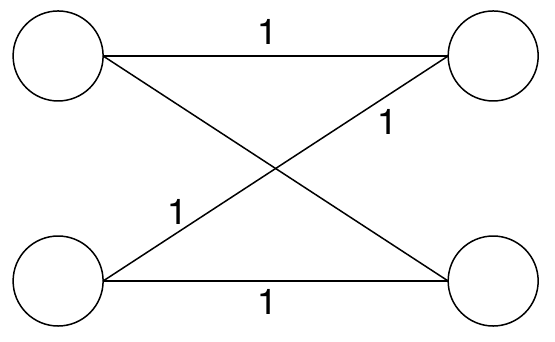}
    \label{fig:NE:nonunique-1}
    \caption{}
\end{subfigure}
\begin{subfigure}{0.3\textwidth}
    \includegraphics[width=\textwidth]{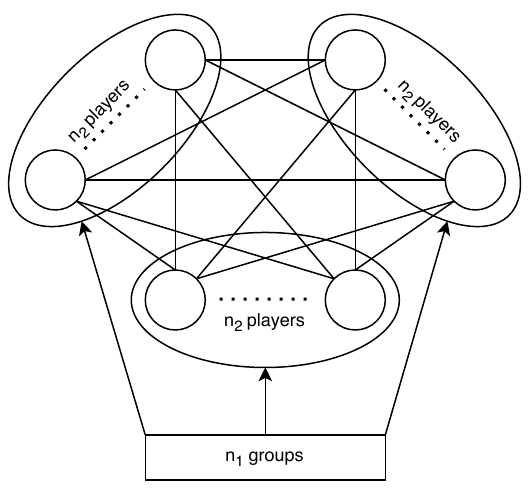}
    \label{fig:NE:nonunique-2}
    \caption{}
\end{subfigure}
\caption{Examples of non-unique NE in public good games. (a): There are four players on two sides, with two players in each side. (b) There are $n_1\times n_2$ players in $n_1$ groups, with $n_2$ players in each group.}
\label{fig:NE:nonunique}
\end{figure}

\begin{example}
\label{eg:NE:nonunique}
Consider a public good game containing four players, see \cref{fig:NE:nonunique}-(a) . The marginal gain is $1$ between one player from the left side and the other from the right side, and $0$ otherwise. We specify homogeneous utility functions and action spaces for all players. The action space is specified as $[0,1]$, while the only two constraints for utility functions are: 
\begin{align*}
    f'(1) \ge c'(1)~~\mbox{and}~~
    f'(2) \le c'(0).
\end{align*}
It's straightforward to verify that players on one side exert full effort, i.e., $x_i = 1$,
while players on the opposite side free ride, i.e., $x_i = 0$, constitutes a Nash Equilibrium (NE). Thus, there are at least two NEs in this game. This example can be readily extended to a scenario involving $n_1\times n_2$ players, distributed into $n_1$ groups with $n_2$ players in each group. All pairs of players from different groups are connected, see \cref{fig:NE:nonunique}-(b). The second condition then becomes $f'(n_2) \le c'(0)$.
\end{example}

Example \ref{eg:NE:nonunique} motivates us to investigate the scenarios in which the NE is unique.

\subsection{Uniqueness of Nash Equilibrium}
\label{subsec:NE:unique}

In this section, we explore the conditions under which NE of the public good game is unique.
We begin by introducing a necessary lemma along with some definitions that will be frequently utilized in the subsequent theorems.

\begin{definition}[$(\gamma, \sigma)$-closeness]
\label{def:closeness}
    A function $g(x): X \to \bbR$ is $(\gamma, \sigma)$-close ($\gamma, \sigma \in \bbR_{+}$) to a function $f(x): X \to \bbR$, where $X \subseteq \bbR^k$, if $\gamma \nabla_x g(x) - \nabla_x f(x)$ is $\sigma$-Lipschitz on $x$.
\end{definition}

\begin{definition}[Near-potential Game]
\label{def:near-potential}
    Consider a game containing $n$ players. Denote $x_i \in X_i \subset \bbR$ as the action of player $i$ and $u_i(\bmx)$ as the utility function of player $i$ given the joint action $\bmx$. 

    We say that a game is a $(\bmgg,\Sigma)$-near-potential ($\bmgg \in \bbR_{++}^n, \Sigma\in\bbR_+^{n\times n}$) game \wrt~  a potential function $u(\bmx)$, if for two players $\ijinn$ (it could be $i = j$), we have that $u_i(\bmx)$ is $(\gamma_i, \sigma_{ij})$-close to $u(\bmx)$ on the domain $X_j$, assuming that $\bmx_{-j}$ is fixed, \ie, 
    \begin{align*}
        \gamma_i \frac{\pp u_i}{\pp x_i}(x_j,\bmx_{-j}) - \frac{\pp u}{\pp x_i}(x_j,\bmx_{-j})
    \end{align*}
    is $\sigma_{ij}$-Lipschitz on $x_j$ for all fixed $\bmx_{-j}\in X_{-j}$, where $\Sigma = \{\sigma_{ij}\}_\ijinn$.
\end{definition}

\begin{restatable}{lemma}{lemNearPotential}
\label{lem:near-potential}
For a $(\bmgg,\Sigma)$-near-potential game  \wrt ~potential function $u(\bmx)$, if the following conditions hold
\begin{itemize}
    \item[(1)] $u(\bmx)$ is $c$-strongly concave on $\bmx$;
    \item[(2)] $c > \sigma_{max}(\Sigma)$,
\end{itemize}
where $\sigma_{max}(\cdot)$ represents the maximum singular value of a matrix,
then the near-potential game has a unique NE $\bmx^*$.
Moreover, the $\bmgg$-scaled pseudo-gradient ascent dynamic $\bmx(t)$ with arbitrary initial point $\bmx(0)$ converges to the NE with an exponential rate, \ie, there is $c_0>0$ such that
\begin{align*}
    \| \bmx(t) - \bmx^* \| = O(\exp(-c_0\cdot t)).
\end{align*}

\end{restatable}

Lemma \ref{lem:near-potential} can be deduced from the results in \citet{concave_game-initial:rosen1965existence} by the utilization of \emph{concave games} and \emph{diagonal strict concavity}.
This deduction relies on the technical assumption of the second-order differentiability of $u_i(\bmx)$'s and $u(\bmx)$. 
The proof for \cref{lem:near-potential} can be done by verifying whether the conditions in \citet{concave_game-initial:rosen1965existence} are satisfied, given the conditions in this lemma. We provide a proof sketch below, with the full derivation available in arXiv ???.

\begin{proof}[Proof Sketch of \cref{lem:near-potential}]
\citet{concave_game-initial:rosen1965existence} proved that diagonal strict concavity indicates the uniqueness of NE. He also proposed a sufficient condition for diagonal strict concavity is that $G(\bmx,\bmgg) + G^T(\bmx, \bmgg)$ is negative definite. Here, $G(\bmx, \bmgg)$ is the Jacobian of $g(\bmx, \bmgg)$ \wrt\ $\bmx$, $G^T$ is the transpose of matrix $G$, and $g(\bmx, \bmgg)$ is the vector $(\gamma_i \frac{\pp u_i}{\pp x_i}(\bmx))_\iinn$, representing the pseudo-gradient of game $(u_i(\bmx))_\iinn$.

By careful computation, we can express $G(\bmx, \bmgg)$ as
\begin{align*}
    G(\bmx, \bmgg) =& H(\bmx) + \begin{bmatrix}
    \frac{\pp^2 (\gamma_1 u_1 - u)}{\pp x_1^2} (\bmx) & \cdots & \frac{\pp^2 (\gamma_1 u_1 - u)}{\pp x_1 \pp x_n}(\bmx)
    \\
    \vdots & \ddots & \vdots 
    \\
    \frac{\pp^2 (\gamma u_n - u)}{\pp x_n x_1} (\bmx) & \cdots & \frac{\pp^2 (\gamma u_n - u)}{\pp x_n \pp x_n}(\bmx)
    \end{bmatrix}
    \\
    \triangleq& H(\bmx) + I(\bmx,\bmgg) 
\end{align*}
where $H(\bmx)$ is the Hessian matrix of $u(\bmx)$ \wrt\ $\bmx$, thus is $c$-negative definite. 
By the near-potential property of the game $(u_i(\bmx))_\iinn$, we can bound the $I(\bmx, \bmgg)$ by $\Sigma$, with the largest eigenvalue of $\Sigma + \Sigma^T$ less than $2c$. Therefore, $G(\bmx, \bmgg) +  G^T(\bmx, \bmgg)$ is negative definite, which completes the proof.
\end{proof}







Next, we will present three results of the uniqueness of NE under different conditions.

\begin{theorem}
\label{thm:NE:unique:near-individual}
Given a public goods game $G = \langle \{f_i(k)\}_\iinn,$ $\{c_i(x)\}_\iinn, \{X_i\}_\iinn, W \rangle$. If the following conditions hold, \myx{W is small}
\begin{itemize}
    \item[(1)] $\gamma_i \left( f_i(x + d) - c_i(x) \right)$ is $c$-concave on $x$, for all $i$ and any fixed $d \in [\dlow_i, \dhigh_i] \coloneqq D_i$, where $\dlow_i$ and $\dhigh_i$ are the minimum and maximum externality gains of player $i$, respectively; \footnote{Similar with $\klow_i$ and $\khigh_i$, we have the explicit formula as follows: $\dlow_i = \sum_\jnei \one\{w_{ij} > 0\} w_{ij} \xlow_j + \one\{w_{ij} < 0\} w_{ij} \xhigh_j$ and $\dhigh_i = \sum_\jnei \one\{w_{ij} > 0\} w_{ij} \xhigh_j + \one\{w_{ij} < 0\} w_{ij} \xlow_j$}
    \item[(2)] $f'_i(k)$ is $L_0$-Lipschitz on $k$ for all $i$;
    \item[(3)] $c > L_0 \sigma_{max}(\Sigma)$, where $\Sigma = \{\sigma_{ij}\}_\ijinn$ and $\sigma_{ij} = \sum_\knei \gamma_k | w_{ki}w_{kj} |$,
\end{itemize}
then, the NE is unique.
\end{theorem}

\begin{proof}
\label{prf:thm:NE:unique:near-individual}
We shall apply \cref{lem:near-potential} to prove this theorem.

Firstly, we construct a near-potential game by  specifying the potential function $u(\bmx)$ and the utilities $u_i(\bmx)$ for all players.

To do this, we let the potential function $u(\bmx) = \sum_\iinn \gamma_i \left( f_i(k_i) - c_i(x_i) \right)$, and specify the utilities $u_i(\bmx)$ in the near-potential game identical to the utilities in the public good game. With some straightforward calculations, we derive that
\begin{align*}
    \frac{\pp u_i}{\pp x_i}(\bmx) =& f'_i(k_i) - c'_i(x_i);
    \\
    \frac{\pp u}{\pp x_i}(\bmx) =& \sum_{i' \ne i} \gamma_{i'} f'_{i'}(k_{i'})w_{i'i} + \gamma_i \left( f'_i(k_i) - c'_i(x_i) \right).
\end{align*}

Since $f'_{i'}(k_{i'})$ is $L_0$-Lipschitz on $k_{i'}$, and $k_{i'}$ is $|w_{i'j}|$-Lipschitz on $x_j$, we have $\gamma_{i'} f'_{i'}(k_{i'})$ is $L_0 \gamma_{i'} |w_{i'j}|$-Lipschitz on $x_j$ and $\sum_{i'\ne i} \gamma_{i'} f'_{i'}(k_{i'}) w_{i'i}$ is $ L_0 \sum_{i'\ne i} \gamma_{i'} |w_{i'j} w_{i'i}|$-Lipschitz on $x_j$.

By constructing matrix $\Sigma = \{\sigma_{ij}\}_\ijinn$ with $\sigma_{ij} = L_0 \sum_\knei \gamma_k |w_{ki} w_{kj}|$, we can prove that $u_i(\bmx)$ is $\Sigma$-near-potential respect to $u(\bmx)$.
By \cref{lem:near-potential}, we obtain the result and thus complete the proof.
\end{proof}

\begin{remark}
\label{rmk:NE:unique:near-individual}
The conditions in \cref{thm:NE:unique:near-individual} intuitively means that the players are close to playing an individual-interest game, \ie, the non-diagonal elements of $W$---those describe the interactions among different players---are small enough. In fact, from the expression of potential $u(\bmx) = \sum_\iinn \gamma_i u_i(\bmx)$, we know that the NE solution is close to the (weighted) social optimal solution.
\end{remark}

\begin{restatable}{theorem}{thmNENearPotential}
\label{thm:NE:unique:near-potential}
Given a public goods game $G = \langle \{f_i(k)\}_\iinn,$ $\{c_i(x)\}_\iinn, \{X_i\}_\iinn, W \rangle$. If the following conditions hold, 
\myx{W is near potential}
\begin{itemize}
    \item[(1)] $f_i(k)$ is $(\gamma_i,\sigma_i)$-close to $f(k)$ for all $\iinn$;
    \item[(2)] $f(x + d) - \gamma_i c_i(x)$ is $c$-strongly concave on $x$ for all $\iinn$ and all $d \in [\dlow_i, \dhigh_i]$, and $f'(k)$ is $c^1$-Lipschitz on $k$, $f''(k)$ is $c^2$-Lipschitz on $k$, $c,c^1,c^2\in 
    \bbR_+$;
    \item[(3)] $c > \sigma_{max}(B)$, where $B = \{\beta_{ij}\}_\ijinn$ and $\beta_{ij} = \sigma_i |w_{ij}| + c^1 |w_{ij}-1| + c^2 \sum_\jinn |w_{ij}-1| \max\{ -\xlow_j, \xhigh_j \}$,
\end{itemize}
then the NE is unique.
\end{restatable}

\begin{remark}
\label{rmk:NE:unique:near-potential}
It's important to note that the conditions specified in \cref{thm:NE:unique:near-potential} intuitively suggest that each element of $W$ closely approximates $1$, and the values derived from the gains ${f_i(k)}_\iinn$ are nearly identical (when scaled). Consequently, the game approaches the characteristics of an identical-interest game, where the players' actions nearly maximize the potential function $u(\bmx) = f(\| \bmx \|_1) - \sum_\iinn \gamma_i c_i(x_i)$.
However, the social welfare is close to $n f(\| \bmx \|_1) - \sum_\iinn \gamma_i c_i(x_i)$, the $\frac{1}{n}$ coefficients on values means that in this case, the free-ride phenomenon can occur.
\end{remark}

\begin{restatable}{theorem}{thmNENearSymmetric}
\label{thm:NE:unique:near-symmetric}
Given a public goods game $G = \langle \{f_i(k)\}_\iinn,$ $\{c_i(x)\}_\iinn, \{X_i\}_\iinn, W \rangle$. If the following conditions hold, 
\myx{W is near positive definite}
\begin{itemize}
    \item[(1)] $W^0$ is positive definite and $\sigma_{min}(W^0) = \sigma_0 > 0$. We also restrict $w^0_{ii} = 1,\ \forall \iinn$ where $W^0 = \{w^0_{ij}\}_\ijinn$;
    \item[(2)] $c'_i(x)$ is $L_i$-Lipschitz on $x$ for all $i$;
    \item[(3)] $f_i(k)$ is $C_i$-concave on $k$ for all $i$;
    \item[(4)] $\sigma_0 > \sigma_{max}(\Sigma)$, where $\Sigma = \{\sigma_{ij}\}_\ijinn$ and $\sigma_{ii} = 0$ and $\sigma_{ij} = \frac{2L_i |w_{ij}|}{C_i} + |w^0_{ij} - w_{ij}|$,
\end{itemize}
where $\sigma_{min}(W)$ represents the minimal eigenvalue of a symmetric matrix $W$, then the NE is unique.
\end{restatable}

\begin{proof}[Proof Sketch of \cref{thm:NE:unique:near-symmetric}]
\citet{public-network-direct-BRD:bayer2023best} proved that, when $W$ is symmetric and the cost functions $c_i(x)$s are linear, then the best-response dynamic converges. The insight is that when $c_i(x)$s are linear, each player $i$ has its own marginal cost $c_i$, and the ideal $k_i$ such that $f'_i(k_i) = c_i$. Therefore, each player $i$ plays the best response to her ideal gain $k_i$, and $\phi(\bmx) = \bmk^T \bmx - \frac{1}{2} \bmx^T W \bmx$ becomes a potential function. Moreover, the NE must be unique if $W$ is positive semi-definite.

Our proof follows this insight. We construct the potential function $\phi(\bmx) = \bmk^{*T}\bmx - \frac{1}{2}\bmx^T W^0 \bmx$. Similarly define $y_i(\bmx_{-i})$ as the optimal gain level of player $i$, when the strategy profile of other players is $\bmx_{-i}$. Then the utilities in the near-potential game are, 
\begin{align*}
    \phi_i(\bmx) = y_i(\bmx_{-i}) x_i - \frac{x_i^2}{2} - \sum_\jnei w_{ij} x_i x_j.
\end{align*}
We then prove that: (1) the NE of the near-potential game corresponds to the NE of the original public good game; and (2) the constructed game $\{\phi_i(\bmx)\}_\iinn$ is indeed a near-potential game.
Following these results, the proof can be completed by 
\cref{lem:near-potential}.
\end{proof}

\begin{remark}
Theorem \ref{thm:NE:unique:near-symmetric} intuitively suggests that, if $W$ is close to a positive definite matrix $W_0$, as well as that the profit functions $f_i(k)$s are more concave than cost functions $c_i(x)$s, then the NE is unique.
\end{remark}

In addition to these three theorems that establish the uniqueness of the NE, we introduce a concept, called \emph{game equivalence}, which can expand the applicability of these theorems.

\begin{definition}[Game Equivalence]
\label{def:equivalence}
Given two public goods games $G^1, G^2$ with $n$ players, where 
\begin{align*}
    G^j = (\{f^j_i\}_\iinn, \{c^j_i\}_\iinn, \{X^j_i = [\xlow^j_i, \xhigh^j_i]\}_\iinn, W^j), j \in \{1,2\}.
\end{align*}
$G^1$ is equivalent to $G^2$, if there is a diagonal matrix $D = \diag(d_1,...,d_n)$, $d_i\in \bbR_{++}$ and an offset vector $\bm{b} \in \bbR^n$, satisfying that, 
\begin{align*}
    W^2 =& D W^1 D^{-1};
    \\
    \xlow^2_i =& d_i \xlow^1_i + b_i;
    \\
    \xhigh^2_i =& d_i \xhigh^1_i + b_i;
    \\
    c^1_i(x) =& c^2_i(d_i x + b_i)\quad\forall x \in X^1_i;
    \\
    f^1_i(k) =& f^2_i(d_i k + m_i)\quad\forall k \in K^1_i,
\end{align*}
where $m_1,...,m_n$ are constants such that $m_i = d_i \sum_\jinn \frac{w^1_{ij} b_j}{d_j}$.

\end{definition}

Intuitively, Definition \ref{def:equivalence} states that if $G^1$ is equivalent to $G^2$, then $G^1$ and $G^2$ are intrinsically the same in terms of linear transformation. Through this insight, we have the following theorem.

\begin{restatable}{theorem}{thmNEEquivalence}
\label{thm:NE:unique:equivalence}
If two games, $G_1$ and $G_2$, are equivalent, then there exists a one-to-one mapping between NEs of $G_1$ and the NEs of $G_2$.
\end{restatable}

From \cref{thm:NE:unique:equivalence},it is evident that the uniqueness property of the NE is preserved within the equivalent class. Therefore, we present the following corollary, which further broadens the class of public goods games that have a unique NE.

\begin{corollary}
\label{cor:NE:unique:equivalence}
For a public goods game $G^1$, if $G^1$ is equivalent to the  game $G^2$, and $G^2$ satisfies the conditions in \cref{thm:NE:unique:near-individual}, \cref{thm:NE:unique:near-potential} or \cref{thm:NE:unique:near-symmetric}, then $G^1$ has a unique NE.
\end{corollary}

\section{Case Study}
\label{sec:case}

\subsection{Comparative Statics: Money Redistribution for Welfare Analysis}
\label{subsec:comparative_money}


In this section, we study comparative statics, \ie, how the players' utilities will change if the model parameters are modified by an infinitesimal amount. 
We characterize the infinitesimal modification by money redistribution, \ie, replace $\{f_i(k_i)\}_\iinn$ by $\{f_i(k_i + \delta_i t)\}_\iinn$, where $\bmdd = (\delta_1,\dots, \delta_n)\in\bbR^n$ is called the direction of money redistribution and $t\in \bbR$ is called the change magnitude. 
Overall, there is a $\bmdd t$ shift in the gain level of players.
The goal of infinitesimal change drives us to study the case $t\to 0$.

In this way, the utility of player $i$ becomes
\begin{eqnarray*}
  u_i(\bmx;t) = f_i(k_i + \delta_i t) - c_i(x_i).   
\end{eqnarray*}

Denote $\bmx^*(t)$ as the NE when the change magnitude is $t$. We do not assume the uniqueness of NE anymore, and $\bmx^*(t)$ might be not unique. However, we assume the first-order differentiability of $\bmx^*(t)$ with respect to $t$, as well as that $\bmx^*(t)$ is an inner point of $X$. 
These assumptions are quite natural. 
For the first assumption, if the game changes with an infinitesimal magnitude and players always achieve the rational outcome, \ie, NE, then it is imaginable and intuitive that the outcome of players should also change minimally.
The second assumption is only technical.
We denote $u_i(t) = u_i(\bmx^*(t);t)$ for a little abuse of notation when the context is clear. We are mainly concerned about $u'_i(0)$, which means that what the marginal change of $\delta$ would affect the players' utilities. Thus, we have the following result.


\begin{restatable}{theorem}{thmComparativeMoney}
\label{thm:comparative:money}
Assume $u_i(t)$ and $\bmx^*(t)$ are defined above, and denote $\bmx^* = \bmx^*(0)$, $\bmk^* = W \bmx^*$, then,
\begin{align*}
    \bmu'(0) = \diag(\bmf'(\bmk^*)) \cdot \diag(\bmc''(\bmx^*) - \bmf''(\bmk^*)) 
    \\
    \left( \diag(\bmc''(\bmx^*)) - W \diag(\bmf''(\bmk^*))\right)^{-1} \bmdd
\end{align*}
where $\bmu(0)$ represents the utility profile $(u_1(0), u_2(0), \dots, u_n(0))$.
\end{restatable}

We also demonstrate some examples to illustrate the implications of the result in \cref{thm:comparative:money}.
\begin{example}
\label{eg:comparative:money:simple}
Here are some simple cases of \cref{thm:comparative:money}.
\begin{enumerate}
\item If the value function is linear on gain, \ie, $f''_i(k) \equiv 0$, then, it becomes that
\begin{align*}
    \bmu'(0) = \diag(\bmf'(\bmk^*)) \bmdd.
\end{align*}
This result is intuitive, because a linear value function indicates that the NE is unique and is constant since the marginal values for players' efforts are constants and marginal costs only depend on players' strategies. Therefore, the change of money redistribution has a direct change on the utilities.

\item If the cost function is linear on effort, \ie, $c''_i(x) \equiv 0$, then, it becomes that
\begin{align*}
    \bmu'(0) = \diag(\bmf'(\bmk^*)) W^{-1} \bmdd.
\end{align*}

In this case, NE might be not constant and not unique (see \cref{eg:NE:nonunique}). Therefore, the redistribution of money will affect the interactions of players, and thus have an indirect effect on the utilities. Specifically, the indirect effect imposes the inverse of $W$---the matrix that portrays the interactions of players---to the money redistribution $\bmdd$.

\item If we want the money redistribution to be Pareto dominant, \ie, $u'_i(0) \ge 0$ for all players, since the first two diagonal matrices are positive diagonal matrices, the only requirement of $\bmdd$ is:
\begin{align*}
    \left[ \diag(\bmc''(\bmx^*)) - W \diag(\bmf''(\bmk^*)) \right]^{-1} \bmdd \ge \zeros.
\end{align*}

Besides, a linear cost would reduce the requirements to,
\begin{align*}
    W^{-1}\bmdd \ge \zeros.
\end{align*}

\end{enumerate}

\end{example}





\subsection{Some Applications of Results}
\label{subsec:eg}

In this section, we propose a specific example to illustrate how the results regarding the uniqueness of NE can be applied in practice. This example is inspired by \citet{public-dynamic:fershtman1991dynamic}, where the cost functions are modeled as quadratic functions.

Specifically, we assume the homogeneity of players in the public good game $G$, i.e., the values of gains, costs of efforts, and action spaces are identical among players, with differences only in the network structure $W$. Therefore, we use $f(k)$ and $c(x)$ instead of $f_i(k_i)$ and $c_i(x_i)$ to represents values and costs, when the context allows. 

Assume $f(k)$ and $c(x)$ has following expression:
\begin{align*}
    f(k) =& 
    \begin{cases}
    ak - b k^2 \quad &\text{if $0\le k \le \frac{a}{2b}$}
    \\
    \frac{a^2}{4b} \quad &\text{if $k > \frac{a}{2b}$}
    \end{cases}
    \\
    c(x) =& \frac{c_0}{2}x^2\quad\qquad\quad \text{for $c_0>0$}
\end{align*}
and $X = [0,\xhigh]$ for a sufficiently large 
$\xhigh$ such that choosing $\xhigh$ is a dominated strategy for all players, due to extremely high costs and bounded values for gains. 
The values and costs are quadratic functions in their domains, with a clipping on the value function at the maximum point.
We also restrict $w_{ij}$ to be either $0$ or $1$.


From the expressions of $f(k)$ and $c(x)$, we know that $c(x)$ is $c_0$-strongly convex, $c'(x)$ is $c_0$-Lipschitz, $f(k)$ is $2b$-strongly concave in the domain $[0,\frac{a}{2b}]$ , and $f'(k)$ is $2b$-Lipschitz on the full domain.

\subsubsection{The Application of \cref{thm:NE:unique:near-individual}}
In this part, we assume that the non-diagonal elements of $W$ are $\iid$ generated with probability $p = \frac{p_0}{n}$ equals to $1$ and $0$ otherwise, where $p_0>0$ is a constant.
We have the following theorem,
\begin{restatable}{theorem}{thmCaseOne}
\label{thm:case:1}
\label{thm:eg:independent-w}
    if $\frac{c_0}{2b} > 2p_0 + p_0^2 + \sqrt{n(8p_0 + 10p_0^2 + 4 p_0^3)}$, then with probability at least $\frac{1}{2}$, the public good game $G$ has a unique NE.
\end{restatable}

\begin{remark*}
\label{rmk:thm:case:eg}

This proof is done by substituting \cref{thm:NE:unique:near-individual} and using Chebyshev's inequalities. Notice that $\sigma_{max}(\Sigma)$ can be bounded by the $\infty$-norm $\| \Sigma \|_\infty$, which is the maximum row sum of $\Sigma$.
We extract the sum of each row $i$ by $\gamma_i$, using Chebyshev's inequalities to bound the tail of $\gamma_i$ and union bound to control $\| \Sigma \|_\infty = \max_i \gamma_i$.

Notice that the result inevitably has a dependency on the square root of $n$ by Chebshev's inequality. 
Due to dependence between $\sigma_{ij}$ and $\sigma_{ij'}$, we can not directly use concentration inequalities, such as Chernoff's inequality \citep{chernoff:chernoff1952measure}, which can help decrease the dependency to $\log n$. 
However, we believe that the $\mathrm{poly} \log(n)$ dependency can be established, by the intrinsic independence on $\{w_{ij}\}_\ijinn$, which allows for further studies.

\end{remark*}

\subsubsection{The Application of \cref{thm:NE:unique:near-symmetric} and \cref{thm:NE:unique:equivalence}}

In this part, we assume that $W$ has a specific up-triangular structure, \ie, $w_{ij} = 0$ if $i>j$. Next, we will show that under this assumption, the NE of public good game is unique.

\begin{restatable}{theorem}{thmCaseTwo}
\label{thm:case:2}
If $W$ is an up-triangular matrix, \ie, $w_{ij}=0$ for $i>j$, then the public good game $G$ has a unique NE.
\end{restatable}

\begin{remark*}
In such scenarios, the conditions specified in \cref{thm:NE:unique:near-individual,thm:NE:unique:near-potential,thm:NE:unique:near-symmetric} may no longer be satisfied. However, we can employ the technique described in \cref{thm:NE:unique:equivalence} to transform the original game $G$ into another game $G'$
that meets the conditions outlined in \cref{thm:NE:unique:near-symmetric}.

Notice that this game must have a unique NE. It is because the following insight: since $w_{ij} $ for $i>j$ means that the efforts of players with lower identifiers $j$ have no externalities on players with higher identifiers $i$. Therefore, player $n$ is playing an individual-interest game, and thus has an optimal strategy $x_n^*$. Given $x_n^*$ fixed, player $n-1$ can also determine an optimal strategy $x_{n-1}^*$. Overall, each player can determine an optimal strategy in turn, which forms an equilibrium. 
However, our proof can give a stronger result that, if $w_{ij} = O(\varepsilon^{i+1-j})$\footnote{here $\varepsilon$ is a constant used in the proof} for $i>j$, we can also guarantee the uniqueness of NE.
\footnote{It is because after the transformation in the proof, the lower-triangular elements hold to be $O(\varepsilon)$.}

\end{remark*}

\section{Conclusion}
\label{sec:conclusion}


In this paper, we have presented a novel approach to understanding networked public goods games featuring heterogeneous players and convex cost functions. Through rigorous analysis and theoretical explorations, we have expanded the conventional understanding of strategic interactions in public goods provision within networked environments. Our model, which integrates heterogeneous benefits and convex costs, provides a more realistic portrayal of individual contributions and the resultant dynamics compared to traditional models with linear and homogeneous cost structures.

The theoretical insights and methodological contributions of our study on networked public goods games with heterogeneous players and convex costs fill a significant gap in the economic theory and also provide new perspectives for policymakers. Specifically, by understanding the conditions under which Nash Equilibrium can be achieved and sustained, policymakers can better design interventions and incentives in the context of the Internet economy and social networks, that encourage optimal contribution levels to public goods. In future research, it would be valuable to extend this model to consider dynamic environments, where players can adjust their strategies over time. Additionally, incorporating stochastic elements to model uncertainty in player interactions could provide further insights into the robustness of the model in more complex and realistic scenarios.

\begin{acks}
This work was supported by the National Natural Science Foundation of China (No. 12471339 and 62172012). 
Xiaotie Deng is the corresponding author. 
We thank all anonymous reviewers for their helpful feedback.
\end{acks}
\nocite{*}

\newpage
\bibliographystyle{ACM-Reference-Format}
\balance
\bibliography{reference}

\newpage
\onecolumn
\appendix



\section{Omitted Proofs}
\label{app:omitted-prf}

\subsection{Proof of \texorpdfstring{\cref{lem:concave}}{}}
\lemConcave*
\begin{proof}
\label{prf:lem:concave}
The case  $c_0 = 0$ is trivial. 
Consider the case for $c_0 > 0$, if we take maximum on both sides of \cref{eq:concave}, we have
\begin{align*}
    g(x^*) =& \max_{y\in X}\quad g(y)
    \\
    \le& \max_{y\in \bbR^d}\quad g(x) + \langle y-x, \nabla g(x)\rangle - \frac{c_0}{2}\| y - x\|^2
\end{align*}

The maximum of RHS is achieved at $y^* = x + \frac{1}{c_0}\nabla g(x)$, and therefore,
\begin{align*}
    g(x^*) \le& g(x) + \langle y^* - x, \nabla g(x)\rangle - \frac{c_0}{2}\| y^* - x \|^2
    \\
    =& g(x) + \frac{1}{2c_0} \| \nabla g(x) \|^2
\end{align*}
which completes the proof.
\end{proof}

\subsection{Proof of \texorpdfstring{\cref{thm:SW:BRD}}{}}
\thmSWBRD*
\begin{proof}
\label{prf:thm:SW:BRD}
\emph{Case 1:}
We first consider the case that one of the conditions holds. If $c_i(x_i)$ is $c_0$-convex, then we know that $\SW(\bmx)$ is $c_0$-convex on $x_i$. Similarly, if $f_i(k_i)$ is $c_0$-concave, since $k_i$ depends linearly on $x_i$, we also know that $\SW(\bmx)$ is $c_0$-convex on $x_i$. Also if $\SW(\bmx)$ is $c_0$-concave on $x_i$ for all $i$, then $\SW(\bmx)$ is $c_0$-concave on $\bmx$. 

As a property of $c_0$-concave function $f(\bmx)$ and maximum point $\bmx^*$, we have
\begin{equation}
\label{eq:convex:ineq}
    f(\bmx^*) - f(\bmx) \le \frac{1}{2c_0} \|\nabla f(\bmx)\|^2.
\end{equation}

Define the energy function $E(t) = \SW(\bmx^*) - \SW(\bmx(t))$, then 
\begin{align*}
    \frac{\dd E(t)}{\dd t} =-\sum_\iinn \frac{\partial \SW}{\partial x_i}(\bmx(t)) \frac{\dd x_i}{\dd t}(t)
    = -\sum_\iinn \left(\frac{\partial \SW}{\partial x_i}(\bmx(t))\right)^2 = -\|\nabla \SW(\bmx(t))\|^2.
\end{align*}
Since $\SW(\bmx)$ is $c_0$-concave, by \cref*{eq:convex:ineq} we have $\|\nabla \SW(\bmx(t))\|^2 \ge 2c_0 (\SW(\bmx^*) - \SW(\bmx)) = 2c_0 E(t)$. Therefore, we have
\begin{align*}
    \frac{\dd E(t)}{\dd t} \le -2c_0 E(t).
\end{align*}
By standard differential equation analysis, we have $E(t) \le \exp(-2c_0 t)E(0)$. Taking $c = -2c_0$ completes the proof.

\emph{Case 2:}
Next, we consider the general case. Define $J(t) = t(\SW(\bmx^*) - \SW(\bmx(t))) + \frac{1}{2}\| \bmx^* - \bmx(t) \|^2$. We have
\begin{align*}
    \frac{\dd J(t)}{\dd t} =& \SW(\bmx^*) - \SW(\bmx(t))  - t \langle \frac{\pp \SW}{\pp \bmx}(\bmx(t)), \frac{\dd \bmx}{\dd t}\rangle - \langle \bmx^* - \bmx(t), \frac{\dd \bmx}{\dd t}\rangle
    \\
    =& \SW(\bmx^*) - \SW(\bmx(t)) - \langle \bmx^* - \bmx(t), \nabla \SW(\bmx(t))\rangle - t \|\nabla \SW(\bmx(t))\|^2
\end{align*}
By concavity of $\SW(\bmx)$ we have $\SW(\bmx^*) - \SW(\bmx(t)) \le \langle \bmx^* - \bmx(t), \nabla \SW(\bmx(t))\rangle$, then we have $\frac{\dd J(t)}{\dd t} \le 0$. This indicates that
\begin{align*}
    J(t) \le J(0) = \frac{1}{2}\| \bmx^* - \bmx(0) \|^2
\end{align*}
where $J(t) \ge t(\SW(\bmx^*) - \SW(\bmx(t)))$, and therefore
\begin{align*}
    \SW(\bmx^*) - \SW(\bmx(t)) \le \frac{1}{2t}\| \bmx^* - \bmx(0) \|^2
\end{align*}
Taking $c = \frac{\| \bmx^* - \bmx(0) \|^2}{2}$ completes the proof.

\end{proof}


\subsection{Proof of \texorpdfstring{\cref{thm:NE:exist}}{}}
\thmNEExist*
\begin{proof}
\label{prf:thm:NE:exist}

We first consider the case where all $c_i(x_i)$s are $c$-strongly concave for some $c>0$, in which case best responses of players are always unique and continuous. We will generalize the result to the general case in the second step.

\paragraph{Case 1: Strongly Convex Cost.}
Consider the best response function of players: $\BR:
\times_\iinn X_i \to \times_\iinn X_i$, where $\BR_i(\bmx)$ represents the best response of player $i$ given the effort profile $\bmx_{-i}$, omitting the dummy variable $x_i$. 

Consider the utility function of player $i$:
\begin{align*}
    u_i(\bmx) = f_i\left(\sum_\jinn w_{ij}x_j\right) - c_i(x_i)
    \\
    \BR_i(\bmx) = \arg\max_{x'_i} u_i(x'_i, \bmx_{-i})
\end{align*}

Now we will show the continuity of $\BR_i(\bmx)$. We assume the negation holds. It indicates that there are two sequences $\{\bmx^j_k\}_{k\in\bbN_+}, j\in\{1,2\}$ such that $\lim \bmx^1_k = \lim \bmx^2_k$ but $\lim \BR_i(\bmx^1_k) \ne \lim \BR_i(\bmx^2_k)$ or one of the limitations do not exist. If the latter holds, since the compactness of $X$ we can choose a sub-sequence of $\{\bmx^j_k\}$ such that the limitation of $\BR_i(\bmx^j_k)$ exists for $j=1,2$. Therefore we only consider the former case.

Let $d = \| \lim \BR_i(\bmx^1_k) - \lim \BR_i(\bmx^2_k) \|$. 
By optimality of $\BR_i(\bmx^j_k)$, we have
\begin{align*}
    u_i(\BR_i(\bmx^j_k), \bmx^j_{k,-i}) \ge u_i(\BR_i(\bmx^{-j}_k, \bmx^j_{k,-i}) + \frac{c}{2} \| \BR_i(\bmx^j_k) - \BR_i(\bmx^{-j}_k)\|^2,\quad j = 1,2
\end{align*}

Summing up with $j=1,2$, we have
\begin{align*}
    \sum_{j=1}^2 \left[ u_i(\BR_i(\bmx^j_k), \bmx^j_{k,-i}) - u_i(\BR_i(\bmx^{j}_k), \bmx^{-j}_{k,-i}) \right] \ge c \| \BR_i(\bmx^1_k) - \BR_i(\bmx^2_k)\|^2
\end{align*}

Then taking limits of $k\to \infty$, we know that LHS becomes $0$ and RHS becomes $c d^2 > 0$, which leads to a contradiction.




Since $\times_\iinn X_i$ is a bounded convex set, by Brouwer's fixed-point theorem, we know that there exists $\bmx \in X$ such that $\BR(\bmx) = \bmx$, which indicates that $\bmx$ is an NE.

\paragraph{Case 2: General Convex Cost.}

To deal with the case that $c_i(x_i)$ might be not strongly concave, we use the technique of utility reshaping. Specifically, we define another public good game $G^\beta = (\{f_i\}_\iinn, \{c^\beta_i\}_\iinn \{X_i\}_\iinn, W)$ where $\beta > 0$ and $c^\beta(x_i) = c(x_i) + \beta x_i^2$. It's obvious that in a public good game $G^\beta$, the cost functions are $\beta$-strongly concave, and the NE must exist.

We denote an NE of $G^\beta$ as $\bmx^\beta$. The next step is to construct a strategy profile $\bmx$, from $\bmx^\beta$ for $\beta >0$, such that $\bmx$ is an NE of $G$. Notice that a simple limit may not work, since there is no guarantee that $\bmx^\beta$ is continuous with $\beta$, even that $\bmx^\beta$ might be unmeasurable in the usual sense.

To resolve this issue, we notice that $\bmx^\beta \in X$ where $X$ is a compact set. By the Bolzano-Weierstrass theorem, we know that there exists a convergent subsequence $\bmx^{\beta_k} \to \bmx$ for some $\bmx \in X$ and  $ \beta_k \overset{k \to \infty}{\longrightarrow} 0$, $k\in\bbZ_+$.

Finally, we verify the NE property of $\bmx$. Notice that $X$ is compact again, we know that $c^{\beta_k}(x_i)$ converges to $c(x_i)$ consistently, therefore, $u^\beta_i(\bmx) = f_i(k_i) - c^{\beta_k}_i(x_i)$ also converges to $u_i(\bmx)$ consistently. Taking limits on both sides of the following equality,
\begin{align*}
    u^{\beta_k}_i(\bmx^{\beta_k}) =& \max_{x'_i\in X_i} u^{\beta_k}_i(x'_i,\bmx^{\beta_k}_{-i}),
\end{align*}
we achieve that, 
\begin{align*}
    u_i(\bmx) =& \max_{x'_i\in X_i} u_i(x'_i,\bmx_{-i}),
\end{align*}
which indicates that $\bmx$ is an NE of $G$.

\end{proof}

\subsection{The derivation of \texorpdfstring{\cref{lem:near-potential}}{} with the results of concave games in \texorpdfstring{\citet{concave_game-initial:rosen1965existence}}{}}
\lemNearPotential*
\begin{proof}

To prove \cref{lem:near-potential}, we only need to check the \emph{diagonal strictly concavity} of the public good game.

Denote $g(\bmx,\bmgg) = (\gamma_1\cdot \frac{\pp u_1}{\pp x_1}(\bmx),\dots, \gamma_n\cdot \frac{\pp u_n}{\pp x_n}(\bmx))$ and $G(\bmx, \bmgg)$ be the Jacobian of $g(\bmx, \bmgg)$ with respect to $\bmx$.
\citet{concave_game-initial:rosen1965existence} shows that if $(G(\bmx,\bmgg) + G^T(\bmx,\bmgg))$ is negative definite for all $\bmx$ where $G^T$ represents the transpose of $G$, then the original game must be diagonal strictly concave.

Now we compute the $G(\bmx,\bmgg)$ directly.
\begin{align*}
    G(\bmx, \bmgg) = \begin{bmatrix}
    \gamma_1\cdot \frac{\pp^2 u_1}{\pp x_1^2} (\bmx) & \cdots & \gamma_1\cdot\frac{\pp^2 u_1}{\pp x_1 \pp x_n}(\bmx)
    \\
    \vdots & \ddots & \vdots 
    \\
    \gamma_n\cdot \frac{\pp^2 u_n}{\pp x_n x_1} (\bmx) & \cdots & \gamma_n\cdot\frac{\pp^2 u_n}{\pp x_n \pp x_n}(\bmx)
    \end{bmatrix}
\end{align*}

Let $H(\bmx)$ be the Hessian matrix of $u(\bmx)$. We have that
\begin{align*}
    G(\bmx, \bmgg) =& H(\bmx) + \begin{bmatrix}
    \frac{\pp^2 (\gamma_1 u_1 - u)}{\pp x_1^2} (\bmx) & \cdots & \frac{\pp^2 (\gamma_1 u_1 - u)}{\pp x_1 \pp x_n}(\bmx)
    \\
    \vdots & \ddots & \vdots 
    \\
    \frac{\pp^2 (\gamma u_n - u)}{\pp x_n x_1} (\bmx) & \cdots & \frac{\pp^2 (\gamma u_n - u)}{\pp x_n \pp x_n}(\bmx)
    \end{bmatrix}
    \\
    \triangleq& H(\bmx) + I(\bmx,\bmgg) 
\end{align*}

By $c$-concavity of $u(\bmx)$, we know that $H(\bmx) + H^T(\bmx)$ is negative definite with the largest eigenvalue smaller than $-2c$. 

By $(\bmgg,\Sigma)$ near-potential property of the public good game, we know that the term $|\frac{\pp (\gamma_i u_i - u)}{\pp x_i \pp x_j}(\bmx)| \le \sigma_{ij}$. Therefore, the maximum eigenvalue of $I(\bmx,\bmgg) + I^T(\bmx, \bmgg)$ can be upper bounded by the maximum eigenvalue of $\Sigma + \Sigma^T$, which is also upper bounded by $2 \sigma_{max}(\Sigma)$.

Above all, for all $\bmx \ne 0$,
\begin{align*}
    \bmx^T G(\bmx, \bmgg) \bmx
    =& \bmx^T H(\bmx) \bmx + \bmx^T I(\bmx, \bmgg) \bmx
    \\
    \le& -c \| \bmx \|^2 + \| \bmx^T \| \| I(\bmx, \bmgg) \bmx \|
    \\
    \le& -c \| \bmx \|^2 + \| \bmx \| \sigma_{max}(I(\bmx, \bmgg)) \| \bmx \|
    \\
    \le& -c \| \bmx \|^2 + \| \bmx \| \sigma_{max}(\Sigma) \| \bmx \|
    \\
    =& (\sigma_{max}(\Sigma) - c) \| \bmx \|^2
    \\
    < 0
\end{align*}

Therefore, the public good game is diagonal strictly concave, and we completes the proof.

\end{proof}

\subsection{A Self-contained Proof of \texorpdfstring{\cref{lem:near-potential}}{}}
\begin{proof}
\label{prf:lem:near-potential}

\emph{Step 1: The uniqueness of NE.}

We prove the uniqueness of NE mainly by constructing a contraction mapping and using Banach fixed-point theorem.
The contraction mapping is constructed by a discretized version of the best-response dynamic:
\begin{align*}
    x'_i = g_i(\bmx) = x_i + \varepsilon \gamma_i \frac{\pp u_i}{\pp x_i}(x_i, \bmx_{-i}),\quad \forall \iinn
\end{align*}
for some $\varepsilon > 0$ small enough. Now we prove that $g: X \to X$ is a contraction mapping.

Consider two effort profiles $\bmx, \bmy \in X$, we have that
\begin{align*}
    \| g(\bmx) - g(\bmy) \|^2 =& \sum_\iinn \left( g_i(\bmx) - g_i(\bmy) \right)^2
    \\
    =& \sum_\iinn \left( x_i + \varepsilon\gamma_i \frac{\pp u_i}{\pp x_i}(x_i, \bmx_{-i}) - y_i - \varepsilon\gamma_i \frac{\pp u_i}{\pp x_i}(y_i, \bmy_{-i}) \right)^2
    \\
    =& \| \bmx - \bmy \|^2 + \varepsilon \langle \bmx - \bmy, \gamma_i\left(\frac{\pp u_i}{\pp x_i}(x_i,\bmx_{-i}) - \frac{\pp u_i}{\pp x_i}(y_i,\bmy_{-i})\right)_\iinn\rangle
    \\
    +& \varepsilon^2 \sum_\iinn\gamma_i^2 \left( \frac{\pp u_i}{\pp x_i}(x_i,\bmx_{-i}) - \frac{\pp u_i}{\pp x_i}(y_i,\bmy_{-i}) \right)^2
\end{align*}

We focus on the $\varepsilon$ term:
\begin{align*}
    & \langle \bmx - \bmy, \gamma_i \left( \frac{\pp u_i}{\pp x_i}(x_i,\bmx_{-i}) - \frac{\pp u_i}{\pp x_i}(y_i,\bmy_{-i}) \right)_\iinn \rangle
    \\
    =& \langle \bmx - \bmy, \frac{\pp u}{\pp \bmx}(\bmx) - \frac{\pp u}{\pp \bmx}(\bmy) \rangle
    + \langle \bmx - \bmy, \left( \frac{\pp (\gamma_i u_i - u)}{\pp x_i}(x_i,\bmx_{-i}) - \frac{\pp (\gamma_i u_i - u)}{\pp x_i}(y_i,\bmy_{-i}) \right)_\iinn \rangle
\end{align*}
Since the $c$-concavity of $u(\bmx)$, we have
\begin{align*}
    \langle \bmx - \bmy, \frac{\pp u}{\pp \bmx}(\bmx) - \frac{\pp u}{\pp \bmx}(\bmy) \rangle \le -c \| \bmx - \bmy \|^2
\end{align*}
By $\sigma_{ij}$-Lipschitzness of $\frac{\pp (\gamma_i u_i - u)}{\pp x_i}(\bmx)$ on $x_j$, we have
\begin{align*}
    & \langle \bmx - \bmy, \left( \frac{\pp (\gamma_i u_i - u)}{\pp x_i}(x_i,\bmx_{-i}) - \frac{\pp (\gamma_i u_i - u)}{\pp x_i}(y_i,\bmy_{-i}) \right)_\iinn \rangle 
    \\
    \le& \sum_\iinn \left( \frac{\pp (\gamma_i u_i - u)}{\pp x_i}(x_i,\bmx_{-i}) - \frac{\pp (\gamma_i u_i - u)}{\pp x_i}(y_i,\bmy_{-i}) \right) |x_i - y_i|
    \\
    \le& \sum_\iinn \sum_\jinn \sigma_{ij} |x_j - y_j| |x_i-y_i|
    \\
    =& \| \bmx - \bmy \|^2 \sum_\iinn \sum_\jinn \sigma_{ij} z_i z_j
    \\
    =& \| \bmx - \bmy \|^2 \bmz^T \Sigma \bmz
    \\
    \le& \sigma_{max}(\Sigma) \| \bmx - \bmy \|^2
\end{align*} 

where in the third equality, $z_i \coloneqq \frac{|x_i - y_i|}{\| \bmx - \bmy\|}$, we have $\| z_i \| = 1$. Also notice that $\sigma_{max}(\Sigma) = \max_{\|\bmz\| = \|\bmz'\| = 1} \bmz^T \Sigma \bmz'$.

Therefore, the $\varepsilon$ term:
\begin{align*}
    & \langle \bmx - \bmy, \gamma_i\left( \frac{\pp u_i}{\pp x_i}(x_i,\bmx_{-i}) - \frac{\pp u_i}{\pp x_i}(y_i,\bmy_{-i}) \right)_\iinn \rangle
    \\
    \le& (\sigma_{max} - c) \| \bmx - \bmy \|^2
\end{align*}

Therefore, we can choose $\varepsilon$ so small such that $\| g(\bmx) - g(\bmy) \|^2 \le (1 + \frac{\varepsilon}{2}(\sigma_{max} - c)) \| \bmx - \bmy \|^2 < \| \bmx - \bmy \|^2$ for all $\bmx,\bmy$, which indicates that $g$ is a contraction mapping. By Banach fixed-point theorem, we know that there exists a unique fixed point $\bmx^*$ of $g$, which means that there is a unique $\bmx^*$ such that $\frac{\pp u_i}{\pp x_i}(\bmx^*) = 0$ for all $i$, indicating that $\bmx^*$ is the unique NE.

\emph{Step 2: The exponential convergence rate.}

To do this, we aim at constructing an energy function $E(\bmx)$ such that it holds the following properties:
\begin{itemize}
    \item $E(\bmx) \ge 0$ for all $\bmx$, the equality holds if $\bmx = \bmx^*$.
    \item $\frac{\dd E(\bmx(t))}{\dd t} \le -p_0 E(\bmx(t))$ for some $c_0 > 0$.
\end{itemize}

As long as these properties hold, we immediately get that $E(\bmx(t)) \le E(\bmx(0)) \exp(-p_0 t)$. Then taking $c_0 = p_0$ completes the proof.

We first define the energy function $E(\bmx) = u(\bmx^*) - u(\bmx) + \langle \bmx - \bmx^* , \nabla u(\bmx^*) \rangle$. Since $u(\bmx)$ is a concave function, we know that,
\begin{align*}
    u(y) - u(x) \le \langle y-x, \nabla u(x)\rangle
\end{align*}
Take $x = \bmx^*$ and $y = \bmx$, we derive that $E(\bmx) \ge 0$ for all $\bmx$. When $\bmx = \bmx^*$, we have $E(\bmx^*) = 0$. We also know that $E(\bmx)$ is $c$-strongly convex function.

By little computation and define $v_i(\bmx) = \gamma_i u_i(\bmx) - u(\bmx)$, we have
\begin{align*}
    \frac{\pp E}{\pp \bmx}(\bmx) =& \left(- \frac{\pp u}{\pp \bmx}(\bmx) + \frac{\pp u}{\pp \bmx}(\bmx^*)\right)
    \\
    \frac{\dd x_i}{\dd t}(t) =& \gamma_i \frac{\pp u_i}{\pp x_i}(\bmx(t))
    = \frac{\pp v_i}{\pp x_i}(\bmx(t)) + \frac{\pp u}{\pp x_i}(\bmx(t))
    = \frac{\pp v_i}{\pp x_i}(\bmx(t)) + \frac{\pp u}{\pp x_i}(\bmx^*) - \frac{\pp E}{\pp x_i}(\bmx(t))
\end{align*}

Next, we compute the derivative of $E(\bmx(t))$:
\begin{align*}
    \frac{\dd E}{\dd t}(\bmx(t))
   & =\langle \frac{\pp E}{\pp \bmx}(\bmx(t)), \frac{\dd \bmx}{\dd t}(t)\rangle
    \\
    =& - \| \frac{\pp E}{\pp \bmx}(\bmx(t)) \|^2 \cdots\text{first term}
    \\
    +& \langle \frac{\pp E}{\pp \bmx}(\bmx(t)), \frac{\pp u}{\pp \bmx}(\bmx^*) \rangle
    + \sum_\iinn \frac{\pp v_i}{\pp x_i}(\bmx(t)) \frac{\pp E}{\pp x_i}(\bmx(t)) \cdots\text{second term}
    \\
\end{align*}


We also know $\frac{\pp u_i}{\pp x_i}(\bmx^*) = 0$ by definition of NE.
Combining them in the second term, we achieve,
\begin{align*}
    \text{second term} =& \sum_\iinn \frac{\pp E}{\pp x_i}(\bmx(t))\cdot \frac{\pp u}{\pp x_i}(\bmx^*)
    + \sum_\iinn \frac{\pp v_i}{\pp x_i}(\bmx(t)) \cdot \frac{\pp E}{\pp x_i}(\bmx(t))
    \\
    =& \sum_\iinn \frac{\pp E}{\pp x_i}(\bmx(t))\cdot \frac{\pp (u - \gamma_i u_i)}{\pp x_i}(\bmx^*)
    + \sum_\iinn \frac{\pp v_i}{\pp x_i}(\bmx(t)) \cdot \frac{\pp E}{\pp x_i}(\bmx(t))
    \\
    =& - \sum_\iinn \frac{\pp E}{\pp x_i}(\bmx(t))\cdot \frac{\pp v_i}{\pp x_i}(\bmx^*)
    + \sum_\iinn \frac{\pp v_i}{\pp x_i}(\bmx(t)) \cdot \frac{\pp E}{\pp x_i}(\bmx(t))
    \\
    =& \sum_\iinn \frac{\pp E}{\pp x_i}(\bmx(t))\cdot (\frac{\pp v_i}{\pp x_i}(\bmx(t)) - \frac{\pp v_i}{\pp x_i}(\bmx^*))
    \\
\end{align*}

Denote $\Delta v_i = \frac{\pp v_i}{\pp x_i}(\bmx(t)) - \frac{\pp v_i}{\pp x_i}(\bmx^*)$ and $\Delta \bmv = (\Delta v_1,...,\Delta v_n)$, we have,
\begin{align*}
    \text{second term} =& \langle \frac{\pp E}{\pp \bmx}(\bmx(t)), \Delta \bmv \rangle
    \\
    \le & \|\frac{\pp E}{\pp \bmx}(\bmx(t))\| \|\Delta \bmv\|
\end{align*}

We also know that $\frac{\pp v_i}{\pp x_i}$ is $\sigma_{ij}$-Lipschitz on $x_j$, now we consider $\| \Delta v_i\|$,
\begin{align*}
    \|\Delta v_i \| =& \max_{\|z \|=1} \sum_\iinn z_i \Delta v_i
    \\
    \le& \max_{\|z \|=1} \sum_\iinn\sum_\jinn z_i \sigma_{ij} |x_j(t) - x_j^*|
    \\
    \le& \| \bmx(t) - \bmx^* \| \max_{\|z \|=1, \| y\| =1} \sum_\iinn\sum_\jinn \sigma_{ij} z_i y_j
    \\
    =& \| \bmx(t) - \bmx^* \| \sigma_{max}(\Sigma)
    \\
    \le& \frac{\sigma_{max}(\Sigma)}{c} \| \frac{\pp E}{\pp \bmx}(\bmx(t))\|
\end{align*}

Combining these, we have,
\begin{align*}
    \frac{\dd E}{\dd t}(\bmx(t)) \le& -(1-\frac{\sigma_{max}(\Sigma)}{c}) \| \frac{\pp E}{\pp \bmx}(\bmx(t))\|^2
    \\
    \le& -2(c - \sigma_{max}(\Sigma)) E(\bmx(t))
\end{align*}

Take $p_0 = 2(c - \sigma_{max}(\Sigma))$, we complete the proof.

\end{proof}

\subsection{Proof of \texorpdfstring{\cref{thm:NE:unique:near-potential}}{}}
\thmNENearPotential*
\begin{proof}
\label{prf:thm:NE:unique:near-potential}
consider the potential function $u(\bmx) = f(\sum_\iinn x_i) - \sum_\iinn \gamma_i c_i(x_i)$. We have that $u(\bmx)$ is $c$-strongly concave on $x_i$ for all $i$, thus $c$-strongly concave on $\bmx$.

We also derive that,
\begin{align*}
    \frac{\pp u}{\pp x_i}(\bmx) =& f'(\sum_\iinn x_i) - \gamma_i c_i'(x_i)
    \\
    \gamma_i \frac{\pp u_i}{\pp x_i}(\bmx) =& \gamma_i f'_i(x_i + \sum_\jnei w_{ij} x_j) - \gamma_i c'_i(x_i)
    \\
    \frac{\pp (\gamma_i u_i - u)}{\pp x_i}(\bmx) =& \gamma_i f'_i(x_i + \sum_\jnei w_{ij} x_j) - f'(\sum_\iinn x_i)
    \\
    =& \gamma_i f'_i(x_i + \sum_\jnei w_{ij} x_j) - f'(x_i + \sum_\jnei w_{ij} x_j)\cdots \text{first term}
    \\
    +& f'(\sum_\jinn w_{ij} x_j) - f'(\sum_\iinn x_i)\cdots \text{second term}
\end{align*}

For the first term, since $f_i(k)$ is $(\gamma_i, \sigma_i)$-close to $f(k)$, we know that 
$\gamma_i f'_i(x_i + \sum_\jnei w_{ij} x_j) - \gamma f'(x_i + \sum_\jnei w_{ij} x_j)$ 
is $\sigma_i |w_{ij}|$-Lipschitz on $x_j$ for all $j$.

For the second term, by Lipschitzness of $f'(k)$ and $f''(k)$ and some computations, consider two points $x_j$ and $x_j + \delta$,
\begin{align*}
    & | \left(f'(\sum_\jinn w_{ij} x_j) - f'(\sum_\iinn x_i)\right) - \left(f'(\sum_\knej w_{ik} x_k + w_{ij}(x_j + \delta)) - f'(\sum_\knej x_k + (x_j + \delta))\right) |
    \\
    =& | \left(f'(\sum_\jinn w_{ij} x_j) - f'(\sum_\jinn w_{ij} x_j + \delta)\right) - \left( f'(\sum_\iinn x_i) - f'(\sum_\iinn x_i + \delta)\right)
    \\
    +& f'(\sum_\jinn w_{ij} x_j + \delta) - f'(\sum_\jinn w_{ij} x_j + w_{ij}\delta) |
    \\
    \le& | \left(f'(\sum_\jinn w_{ij} x_j) - f'(\sum_\jinn w_{ij} x_j + \delta)\right) - \left( f'(\sum_\iinn x_i) - f'(\sum_\iinn x_i + \delta)\right) |\cdots\text{first term}
    \\
    +& | f'(\sum_\jinn w_{ij} x_j + \delta) - f'(\sum_\kinn w_{ik} x_k + w_{ij}\delta) |\cdots\text{second term}
    \\
\end{align*}

The first term is upper bounded by $\delta \sum_\jinn |w_{ij}-1| \max\{ -\xlow_j, \xhigh_j \} c^2$, while the second term is upper bounded by $|w_{ij}-1| \delta c^1$.
Overall, this term is $c^1 |w_{ij}-1| + c^2 \sum_\jinn |w_{ij}-1| \max\{ -\xlow_j, \xhigh_j \}$-Lipschitz on $x_j$.

Above all, $u_i$ is $(\gamma_i, \beta_{ij})$-close to $u$ with respect to $x_j$, where
\begin{align*}
    \beta_{ij} = \sigma_i |w_{ij}| + c^1 |w_{ij}-1| + c^2 \sum_\jinn |w_{ij}-1| \max\{ -\xlow_j, \xhigh_j \}
\end{align*}

Therefore, we complete the proof by \cref{lem:near-potential}.

\end{proof}

\subsection{Proof of \texorpdfstring{\cref{thm:NE:unique:near-symmetric}}{}}
\thmNENearSymmetric*
\begin{proof}
\label{prf:thm:NE:unique:near-symmetric}

\citet{public-network-direct-BRD:bayer2023best} shows that, when $W$ is symmetric and the cost functions $c_i(x)$s are linear, then the best-response dynamic converges. The insight is that when $c_i(x)$s are linear, each player $i$ has its own marginal cost $c_i$, and the ideal $k_i$ such that $f'_i(k_i) = c_i$. Therefore, every player $i$'s best-response to her ideal gain $k_i$, and $\phi(\bmx) = \bmk^T \bmx - \frac{1}{2} \bmx^T W \bmx$ becomes a potential function. Moreover, the NE must be unique if $W$ is positive semi-definite.

Our proof follows this insight and tries to utilize the conclusion of \cref{lem:near-potential}. By \cref{thm:NE:exist} we know that there is an NE $\bmx^*$, let $\bmk^* = W \bmx^*$ be the gain profile in the equilibrium level. We construct the potential following,
\begin{align*}
    u(\bmx) = \bmk^{*T} \bmx - \frac{1}{2} \bmx^T W_0 \bmx
\end{align*}
It's easy to observe that $u(\bmx)$ is $\sigma_0$-strongly concave.

Now consider $y_i(\bmx_{-i})$ as the best response function of player $i$, \ie, $y_i(\bmx_{-i})$ is the gain level $k_i$ such that, it's optimal for player $i$ to choose the effort level $x_i$ such that her gain level becomes $k_i$. 
If we define $u^0_i(k_i,\bmx_{-i})$ is the utility of player $i$ when other players play $\bmx_{-i}$ and the gain level of player $i$ is $k_i$, we can write that,
\begin{align*}
    u^0_i(k_i,\bmx_{-i}) =& f_i(k_i) - c_i(k_i - \sum_\jnei w_{ij} x_j)
    \\
    =& (- c^0_i(k_i, \bmx_{-i})) + f_i(k_i)
\end{align*}
where $c^0_i(k_i, \bmx_{-i}) = c_i(k_i - \sum_\jnei w_{ij} x_j)$ is the cost function of player $i$ (in another form).


We have $c'_i(k_i)$ is $L_i$-Lipschitz on $k_i$ by assumption, therefore, we could easily find that $\frac{\pp^2 c^0_i}{\pp k_i \pp x_j}(k_i,\bmx_{-i})$ is upper bounded by $|w_{ij}| L_i$. 

Together with $f_i(k_i)$ is $C_i$-concave on $k_i$, by \cref{lem:argmax:Lipschitz}, we have that
\begin{align*}
    y_i(\bmx_{-i}) = \argmax_{k_i}\quad u^0_i(k_i,\bmx_{-i}) = ( - c^0_i(k_i, \bmx_{-i})) + f_i(k_i)
\end{align*}
is $\frac{2L_i|w_{ij}|}{C_i}$-Lipschitz on $x_j$.

Now we define the utility function for player $i$ in the near-potential game,
\begin{align*}
    u_i(\bmx) = y_i(\bmx_{-i}) x_i - \frac{x_i^2}{2} - \sum_\jnei w_{ij} x_i x_j
\end{align*}

If $\bmx$ is an NE for the game $\{u_i(\bmx)\}_\iinn$ constructed above, we have that
\begin{align}
    0 =& \frac{\pp u_i}{\pp x_i}(\bmx) = y_i(\bmx_{-i}) - x_i - \sum_\jnei w_{ij} x_j
    \\
    \Rightarrow y_i(\bmx_{-i}) =& x_i + \sum_\jnei w_{ij} x_j
\end{align}
\ie, the choice of $x_i$ will make her gain level to $y_i(\bmx_{-i})$, which is also the optimal gain level of player $i$ given $\bmx_{-i}$ by definition of $y(\bmx_{-i})$, therefore, $\bmx$ is also an NE of the original public good game $G$.

The last step is to show that the NE for the near-potential game $\{u_i(\bmx)\}_\iinn$ is unique. We derive that,
\begin{align*}
    \frac{\pp u}{\pp x_i}(\bmx) =& k^*_i - \sum_\jinn w^0_{ij} x_j
    \\
    \frac{\pp u_i}{\pp x_i}(\bmx) =& y_i(\bmx_{-i}) - \sum_\jinn w_{ij} x_j
\end{align*}

It's obvious that $\frac{\pp (u_i - u)}{\pp x_i}(\bmx)$ is $\sigma_{ii} = |w^0_{ii} - 1| = 0$-Lipschitz on $x_i$ (constant on $x_i$) and $\sigma_{ij} = \frac{2L_i|w_{ij}|}{C_i} + |w^0_{ij} - w_{ij}|$-Lipschitz on $x_j$. Thus the constructed utilities, $\{u_i(\bmx)\}_\iinn$, are $(\ones,\Sigma)$-near potential to $u(\bmx)$. By \cref{lem:near-potential}, we have that the NE of the game $\{u_i(\bmx)\}_\iinn$ is unique, which completes the proof.

\end{proof}

\subsection{Proof of \texorpdfstring{\cref{thm:NE:unique:equivalence}}{}}
\thmNEEquivalence*
\begin{proof}
\label{prf:thm:NE:unique:equivalence}

We prove this theorem by reduction, \ie, there is an injective function $g: X^1 \to X^2$ such that if $\bmx^1$ is an NE in $G^1$, then $\bmx^2$ is an NE in $G^2$.

We construct $\bmx^2$ by following,
\begin{align*}
    x^2_i = d_i x^1_i + b_i
\end{align*}
The construction is injective.
Therefore, we have 
\begin{align*}
    k^2_i =& \sum_\jinn w^2_{ij} x^2_j = \sum_\jinn \frac{d_i}{d_j}w^1_{ij} (d_j x^1_j + b_j)
    \\
    =& d_i \sum_\jinn w^1_{ij} x^1_j + d_1 \sum_\jinn \frac{w^1_{ij} b_j}{d_j}
    \\
    =& d_i k^1_i + m_i
\end{align*}

Now fix $\bmx^2_{-i}$, consider the case that player $i$ choose action $x^2_i$:
\begin{align*}
    u^2_i(x^2_i, \bmx^2_{-i}) =& f^2_i(k^2_i) - c^2_i(x^2_i)
    \\
    =& f^2_i(d_i k^1_i + m_i) - c^2_i(d_i x^1_i + b_i)
    \\
    =& f^1_i(k^1_i) - c^1_i(x^1_i)
\end{align*}
which is the maximum utility player $i$ can achieve, since $\bmx^1$ is an NE of $G^1$.
Therefore, $\bmx^2$ is an NE of $G^2$.

We also need to prove that the inverse direction also holds, to clarify this statements, we show that the equivalence relation is symmetric, \ie, if $G^1$ is equivalent to $G^2$, then $G^2$ is also equivalent to $G^1$.

To show this, we let $d'_i = 1/d_i$ and $b'_i = -b_i/d_i$, we have $d'_i \in \bbR_{++}$ and $b'_i \in \bbR$. 
Denote $D' = D^{-1} = \diag(d'_1,...,d'_n)$, then we have
\begin{align*}
    W^1 =& D' W^2 D'^{-1}
    \\
    \xlow^1_i =& d'_i \xlow^2_i + b'_i
    \\
    \xhigh^1_i =& d'_i \xhigh^2_i + b'_i
    \\
    c^2_i(x) =& c^1_i(d'_i x + b'_i)\quad\forall x \in X^2_i
    \\
    f^2_i(k) =& f^1_i(d'_i k + m'_i)\quad\forall k \in K^2_i
\end{align*}
for some constants $\{m'_i\}_\iinn$.
This indicates that an NE of $G^2$ also corresponds to an NE of $G^1$, which completes the proof.

\end{proof}

\subsection{Proof of \texorpdfstring{\cref{thm:comparative:money}}{}}
\thmComparativeMoney*
\begin{proof}
\label{prf:thm:comparative:money}
$\bmx^*(t)$ should satisfy,
\begin{equation}
\label{eq:comparative:money:NE}
    x^*_i(t) = \argmax_{x_i} f_i(x_i + \sum_\jnei w_{ij} x^*_j(t) + \delta_i t) - c_i(x_i),
\end{equation}

By \cref{eq:comparative:money:NE}, we have
\begin{align*}
    c'_i(x^*_i(t)) - f'_i(k^*_i(t) + \delta_i t) = 0,
\end{align*}
which carves out an implicit function
\begin{align*}
    F_i(\bmx^*(t), t) = 0,\quad \forall \iinn
\end{align*}
with $F_i(\bmx,t) = c'_i(x_i) - f'_i(k_i + \delta_i t)$. Take $E_i(t) = F_i(\bmx^*(t), t)$, by implicit function theorem, we have

\begin{equation}
\label{eq:comparative:money:implicit}
    \frac{\dd E_i}{\dd t}(t) = \frac{\pp F_i}{\pp \bmx}(\bmx^*(t),t) \frac{\dd \bmx^*}{\dd t}(t) + \frac{\pp F_i}{\pp t}(\bmx^*(t),t) = 0
\end{equation}

Together \cref{eq:comparative:money:implicit} with all $i$, we have 
\begin{align*}
    \frac{\dd \bmx^*}{\dd t}(t) = - \left( \frac{\pp F}{\pp \bmx} \right)^{-1} \frac{\pp F}{\pp t}(\bmx^*(t),t)
\end{align*}
where $F(\bmx^*(t),t) = (F_1(\bmx^*(t),t), F_2(\bmx^*(t),t), \dots, F_n(\bmx^*(t),t))$, by computation we have,
\begin{align*}
    & \frac{\pp F}{\pp t}(\bmx^*(t),t) = -\diag(\bmf''(\bmk^*(t) + t\bmdd)) \bmdd
    \\
    & \frac{\pp F}{\pp \bmx}(\bmx^*(t),t) = \diag(\bmc''(\bmx^*(t))) - \diag(\bmf''(\bmk^*(t) + t\bmdd))W
\end{align*}

By $\bmu(t) = \bmu(x^*(t);t)$, we derive that,
\begin{align*}
    \bmu'(0) = \frac{\pp \bmu}{\pp t}(\bmx^*;0) + \frac{\pp \bmu}{\pp \bmx}(\bmx^*;0) \frac{\dd \bmx^*}{\dd t}(0)
\end{align*}
where
\begin{align*}
    & \frac{\pp \bmu}{\pp t}(\bmx^*;0) = \diag(\bmf'(\bmk^*)) \bmdd
    \\
    & \frac{\pp \bmu}{\pp \bmx}(\bmx^*;0) = \diag(\bmf'(\bmk^*)) W - \diag(\bmc'(\bmx^*))
    \\
    & \frac{\dd \bmx^*}{\dd t}(0) = \left( \diag(\bmc''(\bmx^*)) - \diag(\bmf''(\bmk^*))W \right)^{-1} \diag(\bmf''(\bmk^*)) \bmdd
\end{align*}

By equilibrium condition, we have
\begin{align*}
    \diag(\bmc'(\bmx^*)) = \diag(\bmf'(\bmk^*))
\end{align*}
and thus
\begin{align*}
    \frac{\pp \bmu}{\pp \bmx}(\bmx^*;0) = \diag(\bmf'(\bmk^*))(W - I)
\end{align*}

Above all,
\begin{align*}
    \bmu'(0) =& \diag(\bmf'(\bmk^*)) \left( \bmdd + (W-I) \left( \diag(\bmc''(\bmx^*)) - \diag(\bmf''(\bmk^*))W \right)^{-1} \diag(\bmf''(\bmk^*)) \bmdd \right)
    \\
    =& \diag(\bmf'(\bmk^*)) \left( I + (W-I) \left( \diag(\bmc''(\bmx^*)) - \diag(\bmf''(\bmk^*))W \right)^{-1} \diag(\bmf''(\bmk^*)) \right) \bmdd
    \\
    =& \diag(\bmf'(\bmk^*)) \left( I + (W-I) \left( \diag(\bmc''(\bmx^*)/\bmf''(\bmk^*)) - W \right)^{-1} \right) \bmdd
    \\
    =& \diag(\bmf'(\bmk^*)) \left( \left( \diag(\bmc''(\bmx^*)/\bmf''(\bmk^*)) - W \right)\cdot \left( \diag(\bmc''(\bmx^*)/\bmf''(\bmk^*)) - W \right)^{-1}
    \right.
    \\
    +& \left. (W-I) \left( \diag(\bmc''(\bmx^*)/\bmf''(\bmk^*)) - W \right)^{-1} \right) \bmdd
    \\
    =& \diag(\bmf'(\bmk^*)) \left( (\diag(\bmc''(\bmx^*)/\bmf''(\bmk^*))-I) \left( \diag(\bmc''(\bmx^*)/\bmf''(\bmk^*)) - W \right)^{-1} \right) \bmdd
    \\
    =& \diag(\bmf'(\bmk^*)) \cdot \diag(\bmc''(\bmx^*) - \bmf''(\bmk^*)) \cdot \left[ \diag(\bmc''(\bmx^*)) - W \diag(\bmf''(\bmk^*)) \right]^{-1} \bmdd
\end{align*}
which completes the proof.

\end{proof}

\subsection{Proof of \texorpdfstring{\cref{thm:case:1}}{}}
\thmCaseOne*
\begin{proof}
\label{prf:thm:case:1}

We firstly substitute the model into the conditions of \cref{thm:NE:unique:near-individual}.
\begin{itemize}
    \item For the first condition, we just specify $\gamma_i \equiv 1$, then $f_i(x+d) - c_i(x)$ is $c$-concave.
    \item For the second condition, $f'_i(k)$ is $2b$-Lipschitz on $k$ for all $i$.
    \item For the third condition, we need that $\sigma_{max}(\Sigma) < \frac{c}{2b}$, where $\sigma_{ij} = \sum_\knei w_{ki}w_{kj}$.
\end{itemize}

It's well-known that $\sigma_{max}(\Sigma) \le \min\{ \| \Sigma \|_\infty , \| \Sigma \|_1 \}$, where $\| \Sigma \|_\infty$ ($\| \Sigma \|_1$) represents the $\infty$-norm (one-norm) of $\Sigma$, \ie, maximum row sum (column sum) of $\Sigma$, respectively. Specifically,
\begin{align*}
    \| \Sigma \|_\infty =& \max_{i} \sum_\jinn \sum_\knei w_{ki}w_{kj}
    \\
    =& \max_i \sum_\knei w_{ki} + \sum_\jnei \sum_\knei w_{ki} w_{kj}
    \\
    =& \max_i \sum_\knei w_{ki} + \sum_\jnei w_{ji} + \sum_\jnei \sum_{k\ne i,j} w_{ki} w_{kj}
    \\
    =& \max_i \quad 2\sum_\jnei w_{ji} + \sum_\jnei \sum_{k\ne i,j} w_{ki} w_{kj}
\end{align*}

Denote $\delta_i = 2\sum_\jnei w_{ij} + \sum_\jnei \sum_{k\ne i,j} w_{ki} w_{kj}$, then,
\begin{align*}
    \bbE_W [\delta_i] = 2(n-1)p + (n-1)(n-2)p^2 \le 2p_0 + p_0^2
\end{align*}
When the context is clear we denote $\bbE_W[\delta] = \bbE_W[\delta_i]$, since this term is constant.

We also compute the second-order moment of $\delta_i$, \ie,
\begin{align*}
    \delta_i^2
    =& \left( 2\sum_{j_1\ne i} w_{j_1,i} + \sum_{j_1\ne i, k_1\ne j_1, i} w_{j_1,i} w_{j_1,k_1} \right)^2
    \\
    =& 4 \sum_{j_1\ne i, j_2\ne i} w_{j_1,i}w_{j_2,i}
    + 4 \sum_{j_2\ne i}\sum_{j_1\ne i, k_1\ne j_1,i} w_{j_1,i}w_{j_1,k_1}w_{j_2,i}
    + \sum_{j_1\ne i, k_1\ne j_1, i}\sum_{j_2\ne i, k_2\ne j_2, i} w_{j_1,i}w_{j_2,i}w_{j_1,k_1}w_{j_2,k_2}
\end{align*}

By simple counting, we have,
\begin{align*}
    \bbE_W[\delta_i^2] = 4(n-1)p + 9(n-1)(n-2)p^2 + 6(n-1)(n-2)^2 p^3 + (n-1)(n-2)(n^2-5n+5) p^4
\end{align*}
and the variance of $\delta_i$,
\begin{align*}
    & \Var_W [\delta_i] = \bbE_W [\delta_i^2] - \bbE_W^2[\delta_i]
    \\
    =& 4(n-1)p + (n-1)(5n-14)p^2 + (n-1)(n-2)(2n-8)p^3 + (n-1)(n-2)(-2n+3)p^4
    \\
    \le& 4p_0 + 5p_0^2 + 2 p_0^3
\end{align*}

Combining them by using Chebyshev inequalities, 
\begin{align*}
    \Pr[\delta_i - \bbE[\delta] \ge k]\le \frac{\Var[\delta_i]}{k^2}
\end{align*}
and
\begin{align*}
    \Pr[\| \Sigma \|_\infty - \bbE[\delta] \ge k]\le \frac{n\Var[\delta_i]}{k^2}
\end{align*}
since $\| \Sigma \|_\infty = \max_\iinn \delta_i$.

To make $\mathrm{RHS} = \frac{1}{2}$, we take $k = \sqrt{n(8p_0 + 10 p_0^2 + 4 p_0^3)}$, therefore, with probability at least $\frac{1}{2}$, we have that $\sigma_{max}(\Sigma) \le \| \Sigma \|_\infty \le 2p_0 + p_0^2 + \sqrt{n(8p_0 + 10 p_0^2 + 4 p_0^3)< \frac{c}{2b}}$, and the game has unique NE by \cref{thm:NE:unique:near-individual}.
    
\end{proof}

\subsection{Proof of \texorpdfstring{\cref{thm:case:2}}{}}
\thmCaseTwo*
\begin{proof}
\label{prf:thm:case:2}
We denote the original game as $G^1 = \{\{f^1_i\},\{c^1_i\},\{X^1_i\},W^1\}$, and the transformed game as $G^2 = \{\{f^2_i\},\{c^2_i\},\{X^2_i\},W^2\}$. $f^1_i$ is $2b$-concave and $(c^{1}_i)'$ is $c$-Lipschitz.

To do a game transformation, we need to construct a scaling vector $\bmd\in\bbR_{++}$ and offset vector $\bmb\in\bbR_{++}$. We set $\bmb = \zeros$ without loss of generality, we have that,
\begin{align*}
    w^2_{ij} = \frac{d_i w^1_{ij}}{d_j}
\end{align*}

$w^2_{ij}$ also equals  $0$ if $i>j$ and $1$ if $i = j$, for $i<j$, we observe that if $d_i << d_j$ as long as $j>i$, then $w^2_{ij}$ can be arbitrarily small. 
In fact, we let $d_i = \varepsilon^{-i}$,
where $\varepsilon>0$ is a pre-specific constant. Therefore, we have that
\begin{align*}
    \begin{cases}
    w^2_{ij} =0\quad &\text{if $i>j$}
    \\
    w^2_{ij} =1\quad &\text{if $i = j$}
    \\
    w^2_{ij} \le \varepsilon &\text{if $i<j$} 
    \end{cases}
\end{align*}

By this transformation, we have that $f^2_i$ is $\frac{2b}{d^2_i}$-concave and $(c^2_i)'$ is $\frac{c}{d^2_i}$-Lipschitz, in \cref{thm:NE:unique:equivalence}.

Now we prove that the game $G^2$ satisfies the conditions in \cref{thm:NE:unique:near-symmetric}. We choose $W^0=I$ so that $\sigma_0 = \sigma_{min}(W^0) = 1$. $L_i = \frac{c}{d^2_i}$ and $C_i = \frac{2b}{d^2_i}$. We compose $\Sigma - \Sigma^1 + \Sigma^2$, where $\Sigma^k = \{\sigma^k_{ij}\}$ and $\sigma^1_{ij} = \frac{2L_i |w^2_{ij}|}{C_i}$, $\sigma^2_{ij} = | w^0_{ij} - w^2_{ij}|$. We have $\sigma_{max}(\Sigma) \le \sigma_{max}(\Sigma^1) + \sigma_{max}(\Sigma^2)$.

Since $|w^0_{ij}-w^2_{ij} |= 0 $ if $i=j$ and $\le \varepsilon$ is $i\ne j$, we know that $\sigma_{max}(\Sigma^2)$ is bounded by $n\varepsilon$. Besides, $\sigma^1_{ij} = \frac{2L_i |w^2_{ij}|}{C_i} \le \frac{c\varepsilon}{b}$, therefore $\sigma_{max}(\Sigma^1)$ is bounded by $\frac{nc\varepsilon}{b}$.
Therefore, if we choose $\varepsilon<\frac{b}{n(b+c)}$, we can get $\sigma_{max}(\Sigma) < 1 = \sigma_0$, which indicates that $G^2$ has a unique NE, which is the same for $G^1$.

\end{proof}

\end{document}